\documentclass[aps,twocolumn,superscriptaddress,notitlepage,pra,showpacs,nofootinbib,10pt]{revtex4-1}
\usepackage{amsmath, amsthm, amssymb}
\usepackage{graphicx}
\usepackage{dcolumn}
\usepackage{bm}
\usepackage{bbm}
\usepackage{pifont}
\usepackage{color}
\usepackage[usenames,dvipsnames]{xcolor}
\usepackage{hyperref}
\usepackage{enumerate}
\usepackage{ulem}
\usepackage{textcomp}
\newcommand{\tmop}[1]{\ensuremath{\operatorname{#1}}}

\newcommand{\tmstrong}[1]{\textbf{#1}}
\newcommand{\tmtextbf}[1]{{\bfseries{#1}}}

\newcommand{\tmverbatim}[1]{{\ttfamily{#1}}}
\newenvironment{enumerateroman}{\begin{enumerate}[i.] }{\end{enumerate}}

\newtheorem{definition}{Definition}
\newtheorem{proposition}{Proposition}

\newcommand{\HermSDP}[1]{\ensuremath{\mathsf{Herm}_+ \!\left( #1 \right)}}
\newcommand{\Sep}[2]{\mathsf{Sep} \left( #1 : #2 \right)}

\newcommand{\tmH}{\mathcal{H}}
\newcommand{\HA}{\mathcal{H}_{\text{A}}}
\newcommand{\HB}{\mathcal{H}_{\text{B}}}
\newcommand{\HX}{\mathcal{H}_{\text{X}}}
\newcommand{\HY}{\mathcal{H}_{\text{Y}}}
\newcommand{\HC}{\mathcal{H}_{\text{C}}}
\newcommand{\HD}{\mathcal{H}_{\text{D}}}
\newcommand{\dA}{d_{\text{A}}}
\newcommand{\dB}{d_{\text{B}}}
\newcommand{\dX}{d_{\text{X}}}
\newcommand{\dY}{d_{\text{Y}}}
\newcommand{\nX}{n_{\text{X}}}
\newcommand{\nY}{n_{\text{Y}}}
\newcommand{\nA}{n_{\text{A}}}
\newcommand{\nB}{n_{\text{B}}}
\newcommand{\rhoAB}{\rho_{\text{AB}}}
\newcommand{\PQ}{P_{\text{Q}}}
\newcommand{\vecPQ}{\vec{P}_{\text{Q}}}
\newcommand{\PSR}{P_{\text{SR}}}

\newcommand{\ket}[1]{\left| #1 \right\rangle}

\newcommand{\ketbra}[1]{\left| #1 \right\rangle \left\langle #1 \right|}

\begin{document}

\title{Practical measurement-device-independent entanglement quantification}
\author{Denis Rosset}
\affiliation{Department of Physics, National Cheng Kung University, Tainan 701, Taiwan}
\affiliation{Perimeter Institute for Theoretical Physics, Waterloo, Ontario, Canada, N2L 2Y5}
\author{Anthony Martin}
\affiliation{Group of Applied Physics, Universit\'e de Gen\`eve, 1211 Gen\`eve, Switzerland}
\author{Ephanielle Verbanis}
\affiliation{Group of Applied Physics, Universit\'e de Gen\`eve, 1211 Gen\`eve, Switzerland}
\author{Charles Ci Wen Lim}
\affiliation{Department of Electrical and Computer Engineering, National University of Singapore, 4 Engineering Drive 3, Singapore 117583}
\author{Rob Thew}
\affiliation{Group of Applied Physics, Universit\'e de Gen\`eve, 1211 Gen\`eve, Switzerland}


\begin{abstract}
  The robust estimation of entanglement is key to the validation of implementations of quantum systems.
  On the one hand, the evaluation of standard entanglement measures, either using quantum tomography or using quantitative entanglement witnesses requires perfect implementation of measurements.
  On the other hand, measurement-device-independent entanglement witnesses (MDIEWs) can certify entanglement of all entangled states using untrusted measurement devices.
  We show that MDIEWs can be used as well to \textit{quantify} entanglement according to standard entanglement measures, and present a practical method to derive such witnesses using experimental data only.
\end{abstract}

\maketitle

Entanglement is a defining feature of quantum theory; entangled quantum systems have an advantage over classical systems in various contexts~{\cite{Horodecki2009}}, among which key distribution~{\cite{Ekert1991}}, quantum computation~{\cite{Jozsa2003}}.
Entanglement exists in various kinds; correspondingly, there exists a variety of entanglement measures to quantify it, even in the simplest case of bipartite quantum systems.
All these measures satisfy the axioms detailed in~{\cite{Horodecki2001}}, mainly that entanglement does not increase under local operations coordinated by classical communication.

Some of these measures have an operational interpretation, such as the distillable entanglement~{\cite{Rains1999a}}, the entanglement cost~{\cite{Hayden2001}}, or the entanglement of formation~{\cite{Wootters1998}}; others can be interpreted geometrically as a distance with respect to the set of separable states, among which are the relative entropy of entanglement~{\cite{Vedral2002}} or the robustness of entanglement~{\cite{Vidal1999,Steiner2003}}; finally, some measures are particularly easy to compute, such as the negativity~{\cite{Vidal2002}} or the recent semidefinite upper bound on distillable entanglement~{\cite{Wang2016}}; a summary is provided in the reviews~\cite{Plenio2007,Eltschka2014}.
Note that the ease of computation is paramount when characterizing experimental realizations of quantum states; some of the entanglement measures can only be evaluated for highly symmetric states, or are NP-hard to compute~\cite{Huang2014}.

In practice, how do we quantify entanglement? Given an unknown bipartite quantum state, we can perform quantum tomography by conducting local measurements, reconstruct the density matrix and then evaluate any of the computable entanglement measures.
While conceptually simple, this procedure suffers from two drawbacks.
First, the reconstruction of a physical state by point-like estimators is always biased and can lead to overestimation of entanglement~{\cite{Schwemmer2015}}.
Second, imperfections in the measurement devices affect entanglement quantification, and also leads to false positives~{\cite{Rosset2012a}}.
The first drawback can be reduced by using proper statistical testing or by linear evaluation~{\cite{Schwemmer2015}}; in the context of entanglement measures, this corresponds to the use of quantitative entanglement witnesses~{\cite{Eisert2007}}.
The second drawback motivated the creation of device-independent and measurement-device-independent methods.

Device-independent methods are built on the following observation.
When the measurements on a bipartite quantum state are space-like separated, the presence of nonlocal correlations certifies the presence of entanglement, and the presence of nonlocal correlations can be tested by using Bell inequalities~{\cite{Bell1964,Brunner2014}}.
Entanglement can even be quantified purely from the correlations~{\cite{Moroder2013}}.
However, device-independent methods have inherent limitations: they do not tolerate arbitrary losses or detection inefficiencies~{\cite{Brunner2014}} and only entangled states with nonlocality can be detected~{\cite{Werner1989,Degorre2005,Barrett2002}}.

Measurement-device-independent entanglement witnesses~{\cite{Branciard2013}} (MDIEWs), based on the semiquantum games introduced in~{\cite{Buscemi2012}}, use a different set of assumptions.
While the measurements need no longer be performed in a space-like separated manner~{\cite{Rosset2013a}}, the measurement devices are driven by quantum inputs whose preparation is trusted.
These witnesses can certify the presence of entanglement in all entangled states, and have been experimentally verified~{\cite{Xu2014,Nawareg2015,Verbanis2016}}.

The maximal payoff of a MDIEW quantifies entanglement.
In~\cite{Shahandeh2017}, it was shown that a single MDIEW is sufficient to quantify negative-partial-transposition (NPT) entanglement.
The measure thus defined has similar properties to the negativity~{\cite{Vidal2002}}, but is not equivalent to it.

However, the question of quantifying entanglement in a measurement-device-independent manner using standard entanglement measures is still open.
In the present work, we generalize the method introduced in our previous Letter~{\cite{Verbanis2016}} by constructing tailored MDIEWs that quantify entanglement.
In contrast with earlier MDIEW constructions~{\cite{Branciard2013,Rosset2013a,Shahandeh2017}}, our approach work directly on the correlations without requiring prior knowledge about the experimental setup, and is robust against changes in the measurement basis or relabelings of measurement outcomes.
Our method is fully general as the resulting MDIEWs can be constructed to provide lower bounds on any convex entanglement measure; moreover, the bound is tight when the measurement devices implement a generalized Bell measurement.
We implement the computations using conic linear programs, which can be efficiently handled by off-the-shelf interior-point solvers.

To demonstrate the general applicability of our method, we use the experimental data collected during our earlier experiment~{\cite{Verbanis2016}}, and compute lower bounds on the negativity, the absolute, random and generalized robustness of entanglement, and the semidefinite upper bound on distillable entanglement due to Wang and Duan~{\cite{Wang2016}}.

\begin{figure*}
  \includegraphics{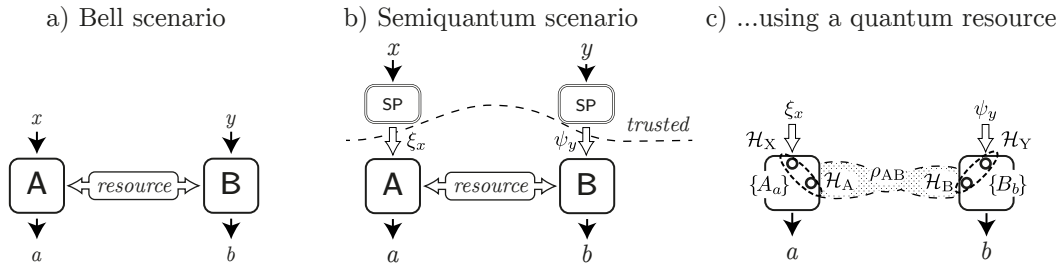}
  \caption{
    \label{Fig:SemiquantumScenario}
    In a), Bell scenarios involve devices that receive classical inputs $x, y$ and give classical outputs $a, b$ after eventual use of a bipartite resource.
    In b), semiquantum scenarios involve trusted state preparation devices (SP), where the indices $x, y$ produce the quantum states $\xi_x$, $\psi_y$ provided to the devices as A and B as quantum inputs; the outputs $a, b$ are still classical.
    In c), an example of a semiquantum scenario where the devices A and B share a quantum state $\rhoAB$, measured jointly with the quantum inputs.}
\end{figure*}

Our paper is structured as follows.
In Section~\ref{Sec:DefinitionsMainClaims}, we recall the main definitions used in the construction of measurement-device-independent entanglement witnesses and define formally their quantitative variant.
In Section~\ref{Sec:Characterizing}, we compute a lower bound on the entanglement of a bipartite state using untrusted measurement devices, with a simple method applicable to tomographically complete sets of quantum inputs.
We also describe optimal measurements that lead to a tight bound.
In Section~\ref{Sec:Quantifying}, we formulate our method as a conic linear program, with two advantages: we relax the experimentally demanding requirement of tomographically complete sets of inputs, and we show how to extract quantitative entanglement witnesses from our formulation.
Finally, we address in Section~\ref{Sec:Experimental} the implementations issues that arise when dealing with losses and noisy experimental data.

\section{Definitions and main claim}
\label{Sec:DefinitionsMainClaims}

The semiquantum scenarios introduced by Buscemi~{\cite{Buscemi2012}} generalize the usual Bell scenarios\footnote{Any Bell scenario can be transformed in a semiquantum scenario by encoding the classical inputs into orthogonal quantum input states.} (see \figurename{~\ref{Fig:SemiquantumScenario}a~\&~\ref{Fig:SemiquantumScenario}b}).
In (bipartite) Bell scenarios, the devices A and B receive measurement settings described by classical inputs $x$ and $y$.
The measurement outcomes $a$ and $b$ are also classical.
In semiquantum scenarios, the devices receive quantum states $\xi_x$ and $\psi_y$ as inputs, taken from the indexed sets $\{ \xi_x \}$ and $\{ \psi_y \}$, but the measurement outcomes are still classical.

In both Bell and semiquantum scenarios, the indices $x$ and $y$ are chosen at random.
In semiquantum scenarios, however, we prepare the state $\xi_x$ (resp. $\psi_y$) from the index $x$ ($y$), and send it to the measurement device.
This state preparation is trusted: we assume that $\xi_x$ ($\psi_y$) is prepared exactly as specified (i.e. a single copy in the proper basis with the prescribed quantum dimension), and that the measurement device receives only the quantum input state $\xi_x$ ($\psi_y$) without the index $x$ ($y$).
The measurement devices A and B process these input states, possibly using a shared resource, and output $a = 1 \ldots \nA$ and $b = 1 \ldots \nB$, respectively. We describe the behavior of the devices by the joint probability distribution $P (a b | x y)$.

Formally, we describe a semiquantum scenario $\mathcal{S}$ by the sets of quantum inputs $\{ \xi_x \}$, $\{ \psi_y \}$ and the number $\nA$, $\nB$ of measurement outcomes:
\begin{equation}
  \label{Eq:SemiquantumScenario}
  \mathcal{S} \equiv \left( \{ \xi_x \}, \{ \psi_y \}, \nA, \nB \right),
\end{equation}
where the input states are described by density matrices $\xi_x \in \mathsf{Herm}_+ \left( \HX \right)$ and $\psi_y \in \mathsf{Herm}_+ \left( \HY \right)$, with $\tmop{tr} [\xi_x] = \tmop{tr} [\psi_y] = 1$ for $x = 1 \ldots \nX$ and $y = 1 \ldots \nY$.
We write $\mathsf{Herm}_+ \left( \tmH \right)$ the set of positive semidefinite Hermitian operators (i.e. with nonnegative eigenvalues).
In the study that follows, we assume a fixed scenario $\mathcal{S}$ (for example, $\{ \xi_x \}$ and $\{ \psi_y \}$ are the six eigenvalues of the Pauli operators $| \pm x \rangle, | \pm y \rangle, | \pm z \rangle$ and $\nA = \nB = 4$).

\subsection{Resources in a semiquantum scenario}

In the given semiquantum scenario $\mathcal{S}$, we allow the devices A and B to access a particular type of resource and describe the behaviors they can exhibit.
We first allow the devices A and B to share a quantum state $\rhoAB \in \mathsf{Herm}_+ \left( \HA \otimes \HB \right)$, as in \figurename{~\ref{Fig:SemiquantumScenario}c}.
The measurement devices can perform a joint measurement on their part of $\rhoAB$ and their quantum input.
The device A performs a joint measurement on $\HX \otimes \HA$ described by the POVM $\{ A_a \}$, while B performs a joint measurement on $\HB \otimes \HY$ described by the POVM $\{ B_b \}$:
\begin{equation}
 \label{Eq:PQ} \PQ (a b | x y) = \tmop{tr} \left[ (A_a \otimes B_b) \left( \xi_x \otimes \rhoAB \otimes \psi_y \right) \right] .
\end{equation}

In the semiquantum setting, we do not restrict the dimension of $\rhoAB$ and make no particular assumptions about the measurements $\{ A_a \}$ and $\{ B_b \}$.
However, when A and B have only access to shared randomness, we obtain
\begin{equation}\label{Eq:PSR} 
	\PSR (a b | x y) = \sum_{\lambda} p_{\lambda}
	\tmop{tr} 
	\left[ \Pi_{a | \lambda}^{\text{A}} \xi_x \right] \tmop{tr} \left[ \Pi_{b | \lambda}^{\text{B}} \psi_y \right],
\end{equation}
where the local hidden variable $\lambda$, distributed according to $p_\lambda$, selects local measurement operators $\Pi_{a | \lambda}^{\text{A}}$ and $\Pi_{b | \lambda}^{\text{B}}$.
As separable states can be created from shared randomness, Eq.~(\ref{Eq:PSR}) also covers the case where A and B share a separable state $\rhoAB \in \Sep{\HA}{\HB}$ --- we write $\Sep{\HA}{\HB}$ for the cone of separable operators in $\HA \otimes \HB$, see Appendix~\ref{App:RecipesSep}.

In addition to shared randomness, we can allow A and B to communicate classical information, and write the resulting correlations $P_{\text{LOCC}} (a b | x y)$.
While the mathematical characterization of the LOCC correlations is complicated~{\cite{Chitambar2014}}, a clever use of entanglement measures will allow us to avoid the problem in Section~\ref{Sec:Characterizing}.

\subsection{Measurement-device-independent entanglement witnesses}

Following~{\cite{Branciard2013}} (albeit with opposite sign convention), we define {\textit{measurement-device-independent entanglement witnesses}} (MDIEWs) by coefficients $\beta_{a b x y}$, so that the expectation value 
\begin{equation}
 \label{Eq:ExpectationValue} I (\vec{P}) \equiv \sum_{a b x y} \beta_{a b x y} P (a b | x y)
\end{equation}
obeys the following requirements:
\begin{enumerateroman}
 \item If A and B only have access to local operations and classical communication, then $I \left( \vec{P}_{\text{LOCC}} \right) \leqslant 0$.
 \item We have $I \left( \vecPQ \right) > 0$ when A and B share a particular entangled state $\rhoAB$ and perform specific joint measurements.
\end{enumerateroman}
In the definition above, we wrote $\vec{P} \in \mathbbm{R}^{\nA \nB \nX \nY}$ as a shorthand for $P (a b | x y)$ by an enumeration of its coefficients.
The robustness of our MDIEW rests on the relation $I \left( \vec{P}_{\text{LOCC}} \right) \leqslant 0$, satisfied for any measurement strategy implemented by the devices, even allowing classical communication~{\cite{Rosset2013a}}.

\subsection{Entanglement measures}
\label{Sec:EntanglementMeasures}
We now construct quantitative MDIEWs.

\begin{definition}
  \label{Def:QuantitativeMDIEW}
 Given entanglement measure $\mathcal{E}$, a measurement-device-independent entanglement witness (MDIEW) $\beta_{a b x y}$ is a quantitative MDIEW when its expectation value~(\ref{Eq:ExpectationValue}) provides a lower bound on the entanglement of the state $\rhoAB$ shared by the devices:
 \begin{equation}
   I (\vec{P}) \leqslant \mathcal{E} \left( \rhoAB \right) \;,
 \end{equation}
 for all states $\rhoAB$ and measurements $\{A_a\}$, $\{B_b\}$.
\end{definition}

Our construction applies to any entanglement measure $\mathcal{E}$ that satisfies the following axioms:
\begin{enumerateroman}
 \item The entanglement measured by $\mathcal{E}$ cannot increase under LOCC operations (axioms M1 and M2 of~{\cite{Horodecki2001}}).
 \item $\mathcal{E}$ is convex (axiom M3.b of~{\cite{Horodecki2001}}).
 \item $\mathcal{E}$ is dimension independent, that is, embedding $\rhoAB$ in a higher-dimensional Hilbert space does not change the amount of entanglement; for separable $\sigma_{\text{A}' \text{B}'} \in \Sep{\tmH_{\text{A}'}}{\tmH_{\text{B}'}}$:
 \begin{equation} \label{Eq:AxiomDimensionIndependence} 
 \mathcal{E} \left( \rhoAB \right) =\mathcal{E} \left( \rhoAB \otimes \sigma_{\text{A}' \text{B}'} \right).
\end{equation}
\end{enumerateroman}
The second requirement allows the use of convex solvers to compute entanglement.
The third requirement allows the interpretation of $\mathcal{E}$ as an entanglement measure when the dimension of $\rhoAB$ is unknown.
We list below common entanglement measures and the axioms they satisfy.

\begin{center}
  \begin{tabular}{|l|l|l|l|l|}
    \hline
    & i.  & ii. & iii.  \\
Negativity~{\cite{Vidal2002}}        & \checkmark & \checkmark & \checkmark   \\
Absolute robustness~{\cite{Vidal1999}}      & \checkmark & \checkmark & \checkmark     \\
Generalized robustness~{\cite{Steiner2003}}   & \checkmark & \checkmark & \checkmark  \\
Random robustness~{\cite{Vidal1999}}        & \checkmark & \checkmark & \ding{55} \\
    Upper bound on distillable ent.~{\cite{Wang2016}} & \checkmark & \checkmark & ? \\
    \hline
  \end{tabular}
\end{center}

To simplify our presentation, we require $\mathcal{E}$ to be invariant under global transposition~\footnote{We leave the following puzzle to the reader.
  Is there any known entanglement measure {\textit{not}} invariant under global transposition?
  In that case, transposes need to be added to POVM elements in our construction.
}:
\begin{equation} \label{Eq:AxiomTranspose} 
  \mathcal{E} \left( \rhoAB \right) =\mathcal{E} \left( \rhoAB^{\top} \right) \;,
\end{equation}
and extend the domain of validity of $\mathcal{E}$ to unnormalized states:
\begin{equation}
  \label{Eq:ExtendValidity}
  \mathcal{E} \left( \rhoAB \right) \equiv \tmop{tr} \left[ \rhoAB \right] \mathcal{E} \left( \rhoAB / \tmop{tr} \left[ \rhoAB \right] \right), \quad \mathcal{E} (0) \equiv 0\;,
\end{equation}
so that $\mathcal{E} \left( \alpha \rhoAB \right) = \alpha ~ \mathcal{E} \left( \rhoAB \right)$ for $\alpha \geqslant 0$.

\subsection{Main claim: quantitative MDIEWs}

We come to our main claim: these quantitative MDIEWs are easily obtained using conic linear programming~{\cite{Luenberger2016}}, starting only from the description $\mathcal{S}$ of the semiquantum scenario~(\ref{Eq:SemiquantumScenario}) and the observations $P (a b | x y)$.
Namely, a lower bound on a given entanglement measure $\mathcal{E} \left( \rhoAB \right)$ comes from the optimal solution $\nu^{\ast} = \min \nu$ of following program:
\begin{equation} \label{Eq:ClaimCLP} 
\arraycolsep=2pt\def\arraystretch{2.2}
\begin{array}{lll}
  \multicolumn{3}{c}{\text{\tmstrong{Entanglement quantification program}}}\\
  \text{minimize} & \nu \equiv \sum_{a b} \mathcal{E}\left(\Pi_{ab}\right) / (\dX \dY) & \\
  \text{over} & \Pi_{a b} \in \HermSDP{\HX \otimes \HY} & \forall a b\\
  \text{subject to} & \tmop{tr} [\Pi_{a b} \cdot (\xi_x \otimes \psi_y)] = P (a b | x y) & \forall a b x y
\end{array}
\end{equation}
where $\{ \Pi_{a b} \}$ is a POVM that effectively describes the behavior of the untrusted devices in the setup.
In Section~\ref{Sec:Characterizing}, we show that these POVM elements are proportional to states that can be recovered from the setup, such that the objective function computes the average entanglement in the recovered ensemble, and is thus a proper lower bound on $\mathcal{E}(\rhoAB)$. 
Reformulating slighlty~\eqref{Eq:ClaimCLP} in Section~\ref{Sec:Quantifying}, we demonstrate that the dual variables $\beta_{a b x y}$ in its numerical solution~{\cite{Sturm2002,Tutuncu2003}} form a quantitative MDIEW\footnote{It it similar to the use of linear programming to check the (non)locality of given correlations; the dual variables provide the particular Bell inequality violated by those correlations~{\cite{Zukowski1999a,Kaszlikowski2000,Brunner2014}}.}.

The program~\eqref{Eq:ClaimCLP} stays valid when the sets of inputs are not tomographically complete.
Moreover, it can be adapted when measurements are not available for all input pairs $(x, y)$; in that case, we simply omit the corresponding equality constraints.

\section{Entanglement in semiquantum scenarios}
\label{Sec:Characterizing}

We now show how to compute a lower bound on the entanglement present in $\rhoAB$ from observable correlations $P(ab|xy)$ when the sets of inputs $\{\xi_x\}$ and $\{\psi_y\}$ are tomographically complete.
To do so, we describe the untrusted part of the semiquantum setup as an effective POVM (\ref{Sec:EffectivePOVM}), and then show that its elements are proportional to bipartite states that can be extracted from the setup (\ref{Sec:EntanglementDistributed}).
The average entanglement of the recovered ensemble provides a lower bound on the entanglement of $\rhoAB$ (\ref{Sec:EntanglementLowerBound}).
We show that the bound is tight when the measurement devices implement Bell measurements (\ref{Sec:OptimalMeasurements}).

\subsection{Semiquantum setups as effective POVMs}
\label{Sec:EffectivePOVM}
For all purposes, our source and measurement devices act together as a black box, distributed between two laboratories A and B.
In each laboratory, the black box receives quantum inputs in $\HX$ (resp. $\HY$) and produces the classical outputs $a$ ($b$); such a system is a joint measurement (\figurename{~\ref{Fig:EffectivePOVM}a}), which we describe using an effective POVM $\{ \Pi_{a b} \in \mathsf{Herm}_+ \left( \HX \otimes \HY \right) \}$ acting on $\xi_x \otimes \psi_y$:
\begin{equation}
  \label{Eq:EffectivePOVMQuantum} 
  \Pi_{a b} = \tmop{tr}_{\tmop{AB}} \left[ \left( \mathbbm{1}_{\text{X}} \otimes \rhoAB \otimes \mathbbm{1}_{\text{Y}} \right) (A_a \otimes B_b) \right] .
\end{equation}
so that Eq.~\eqref{Eq:PQ} writes:
\begin{equation}
 \label{Eq:Tomography} \tmop{tr} [\Pi_{a b} (\xi_x \otimes \psi_y)] = \PQ (a b
 | x y) .
\end{equation}
For simplicity, we assume in the present Section that the sets of inputs are tomographically complete, so that the effective POVM $\{\Pi_{ab}\}$ is completely specified by the linear inversion of Eq.~(\ref{Eq:Tomography}) --- in a sense, a semiquantum experiment performs quantum tomography~{\cite{Lundeen2009}} of an unknown distributed measurement.
This requirement will be relaxed in Section~\ref{Sec:Quantifying}.

\begin{figure}[t]
  \includegraphics[width=\columnwidth]{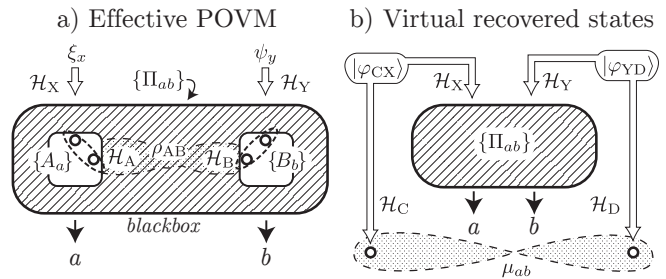}
  \caption{
    \label{Fig:EffectivePOVM}
    a) In a semiquantum setup, the measurement devices A  and B share the state $\rhoAB$ as a resource, and each device performs a joint measurement on the quantum input it receives and part of $\rhoAB$.
    If we see the whole unknown part as a black box, its proper quantum description is that of an effective distributed POVM $\{ \Pi_{a b} \}$. 
    b) As we have no experimental access to $\rhoAB$, we quantify the entanglement present by extracting quantum states from the black box using a (virtual) LOCC protocol. 
  }
\end{figure}

\subsection{Recovering states from the black box}
\label{Sec:EntanglementDistributed}

We consider the recovery of bipartite states $\mu_{a b}$ from the measurement $\{ \Pi_{a b} \}$ by local quantum operations.
For that purpose, we produce maximally entangled states (\figurename{~\ref{Fig:EffectivePOVM}b}) locally in each laboratory.
The first maximally entangled state $\ket{\varphi_{\text{CX}}}$ is produced near the input of the device A, between a reference subsystem $\HC$ and the input subsystem $\HX$. The second maximally entangled state $\ket{\varphi_{\text{YD}}}$ is produced near the input of the device B, between the input subsystem $\HY$ and a reference subsystem $\HD$.
The space $\HC$ (resp. $\HD$) is isomorphic to $\HX$ ($\HY$).
We define, in the computational basis:
\begin{equation}
 \ket{\varphi_{\text{CX}}} \equiv \frac{1}{\sqrt{\dX}} \sum_{i = 1}^{\dX} \ket{i i}, \qquad \ket{\varphi_{\text{YD}}} \equiv \frac{1}{\sqrt{\dY}} \sum_{i = 1}^{\dY} | i i \rangle .
\end{equation}
Outside the black box (and $\rhoAB$), no initial entanglement is present in $\ket{\varphi_{\text{CX}}} \otimes \ket{\varphi_{\text{YD}}}$ between the devices A and B, across the CX/YD boundary.
After performing the measurement $\{ \Pi_{a b} \}$, we obtain the post-measurement states $\mu_{a b}$ with probability $p_{a b}$:
\begin{equation} \label{Eq:RecoveredStates} 
 \mu_{a b} = \frac{\Pi_{a b}^{\top}}{\tmop{tr} [\Pi_{a b}]}, \qquad p_{a b} = \frac{\tmop{tr} [\Pi_{a b}]}{\dX \dY},
\end{equation}
where the transpose is done with respect to the computational basis.
We name those $\mu_{a b}$ the {\textit{(virtual) recovered states}} of the measurement $\{\Pi_{a b} \}$.

\subsection{Entanglement lower bound from recovered states}
\label{Sec:EntanglementLowerBound}
Let us look at the whole process again, including operations performed inside the black box.
We start with a bipartite state $\rhoAB$.
In the laboratory A, we produce a state $\ket{\varphi_{\text{CX}}}$, which is jointly measured with part of $\rhoAB$ using the POVM $\{A_a\}$ implemented by the measurement device. The outputs of this process are the classical measurement outcome $a$ and a quantum state in $\HC$. In laboratory B, the same happens with the production of $\ket{\varphi_{\text{DY}}}$, the measurement of $\{B_b\}$. The outputs are the classical outcome $b$ and a quantum state in $\HD$.

This process is a LOCC operation in the sense of~{\cite{Horodecki2001}}, that transforms the state $\rhoAB$ into the ensemble $\{ p_{a b}, \mu_{a b} \}$ (an ensemble is a set of states with ascribed probabilities).
Such LOCC operations do not increase the entanglement on average~{\cite{Horodecki2001}}.
Consequently, the average entanglement of the ensemble $\{ p_{a b}, \mu_{a b} \}$, which we write $\nu^\ast$, is a lower bound on the entanglement of $\rhoAB$:
\begin{equation} \label{Eq:LowerBound} 
 \nu^\ast \equiv \sum_{a b} p_{a b} \mathcal{E} (\mu_{a b}) \leqslant \mathcal{E} \left( \rhoAB \right) .
\end{equation}
An alternative formulation is obtained by substituting~(\ref{Eq:RecoveredStates}) into~(\ref{Eq:LowerBound}) using requirements~(\ref{Eq:AxiomTranspose}) and~(\ref{Eq:ExtendValidity}):
\begin{equation} \label{Eq:LowerBound1} 
 \nu^{\ast} = \frac{1}{\dX \dY} \sum_{a b} \mathcal{E} (\Pi_{a b}) .
\end{equation}
This last equation is particularly useful as the observations $P (a b | x y)$ provide directly the POVM elements $\{ \Pi_{a b} \}$.
We do not claim that our process provides optimal recovered states $\mu_{a b}$ (and thus optimal $\nu^{\ast}$).
However, it always provides a lower bound on the entanglement of $\rhoAB$.
In certain circumstances, this bound is tight; in particular, when the devices A and B implement optimal measurements $\{ A_a \}$, $\{ B_b\}$ as we see below.

\subsection{Optimal measurements in semiquantum scenarios}
\label{Sec:OptimalMeasurements}
Consider a semiquantum scenario where the parties use a bipartite quantum state $\rhoAB$ as in the correlations~(\ref{Eq:PQ}).
Under which conditions can our procedure detect and quantify entanglement optimally?
Sufficient conditions are provided below.
\begin{proposition}
  \label{Prop:OptimalMeasurements}
  Let $\rhoAB$ be a bipartite state of dimension $\dA \times \dB$.
  Assume that $\rhoAB$ is tested in a semiquantum scenario $\mathcal{S}$ with tomographically complete sets of inputs of dimension $\dX = \dA$, $\dY = \dB$, and that the measurement devices implement generalized Bell measurements $\{ A_a \}$ and $\{ B_b \}$ with $\nA = \dX^2$ and $\nB = \dY^2$ measurement outcomes corresponding to projections on maximally entangled states.  
  Then, the correlations $P (a b | x y)$ lead to a tight bound $\nu^{\ast} =\mathcal{E} \left(\rhoAB \right)$ for any entanglement measure $\mathcal{E}$ in~\eqref{Eq:LowerBound}.
\end{proposition}
\begin{proof}
  See Appendix~\ref{App:OptimalMeasurements}.
\end{proof}

Thus, experimental implementations should strive to implement generalized Bell measurements.
While our method is robust against imperfections in the devices, so that overestimation never happens, eventual imperfections can lead to a lower estimation $\nu^{\ast} <\mathcal{E} \left( \rhoAB \right)$.
The construction presented in this Section can be related to previous works, see Appendix~\ref{App:Relation} for a discussion.
Note that other constructions of quantitative MDIEWs may well be possible, in which case other measurements could be optimal.

\section{Quantifying entanglement by conic programming}
\label{Sec:Quantifying}

We now prove our main claim~\eqref{Eq:ClaimCLP} in two steps.
First, we consider the case where full tomographic data about the setup is not available; then we have to apply the construction of Section~\ref{Sec:Characterizing} by minimizing our bound over all feasible effective POVMs, which is a conic linear program.
Then, we show that the dual of that program provides a quantitative MDIEW.

\subsection{Computing bounds from incomplete data}

We now consider semiquantum scenarios where the sets of inputs $\{\xi_x \}$, $\{ \psi_y \}$ are not necessarily complete.
To reduce the experimental requirements, we also allow incomplete data: say, $P (a b | x y)$ is only available for the pairs of indices $(x, y) \in \mathcal{I}$ for a set $\mathcal{I}$.
Then, the tomography equation~(\ref{Eq:Tomography}) leads to a set of possible solutions $\{ \Pi_{a b} \}$.
We get a lower bound $\nu^{\ast}$ on the entanglement by considering the worst case scenario in~(\ref{Eq:LowerBound1}):
\begin{equation} \label{Eq:MinimizationLowerBound}
 \nu^{\ast} = \min \frac{1}{\dX \dY} \sum_{a b} \nu_{a b},
\end{equation}
\begin{equation} \label{Eq:EntanglementConicConstraint} 
 \mathcal{E} (\Pi_{a b}) \leqslant \nu_{a b}, \qquad \forall a b,
\end{equation}
where the minimization is done over the possible POVMs $\{ \Pi_{a b} \}$ that satisfy~(\ref{Eq:Tomography}), along with the dummy variables $\nu_{a b} \in \mathbbm{R}$.
The constraint~(\ref{Eq:EntanglementConicConstraint}) can be written $(\nu_{a b}, \Pi_{a b}) \in \hat{\mathcal{E}}$ where
\begin{equation}\label{Eq:EntanglementMeasureCone}
 \hat{\mathcal{E}} \equiv \left\{ (\omega, \rho) \text{ such that } \mathcal{E} (\rho) \leqslant \omega \right\}
\end{equation}
is easily verified to be a convex cone (see Appendix~\ref{App:RecipesEntanglement}).
We thus obtain the conic linear form of the program~(\ref{Eq:ClaimCLP}) presented in our introductory claim:
\begin{equation}
  \label{Eq:ClaimCLPDuplicate} 
  \arraycolsep=2pt\def\arraystretch{1.8}
  \begin{array}{ll}
    \multicolumn{2}{c}{\text{\tmstrong{Entanglement quantification }}}  \\[-0.8em]
    \multicolumn{2}{c}{\text{\tmstrong{conic linear program }(primal)}}\\
    \text{minimize } & \nu = \frac{1}{\dX \dY} \sum_{a b} \nu_{a b} \\
    \text{over } & (\nu_{a b}, \Pi_{a b}) \in \hat{\mathcal{E}}, \quad \forall a b\\
    \text{subject to } & \tmop{tr} [\Pi_{a b} \cdot (\xi_x \otimes \psi_y)] = P (a b | x y), \\
    \multicolumn{2}{r}{\forall a b, \quad \forall (x, y) \in \mathcal{I}}.
  \end{array}
\end{equation}
In a semiquantum scenario $\mathcal{S}$, we can solve this program for the observations $P (a b | x y)$, obtain the optimal value $\nu^{\ast}$ which is a lower bound on the entanglement present in $\rhoAB$.
Ready-to-use formulations of $\hat{\mathcal{E}}$ corresponding to various entanglement measures are provided in Appendix~\ref{App:RecipesEntanglement}, in a form that can be directly entered into the toolboxes CVX~{\cite{Grant2014}} and YALMIP~{\cite{Lofberg2004}}.

As part of the solution of the above program, we get the optimal value of the dual variables $\beta^{\ast}_{a b x y}$ corresponding to the equality constraints.
Their interpretation as a quantitative MDIEW is explained below.

\subsection{Quantitative MDIEW from the dual solution}

We now examine how to recover a quantitative MDIEW from the dual program of~(\ref{Eq:ClaimCLPDuplicate}), using the correspondence described in Appendix~\ref{App:CLP}:
\begin{equation} \label{Eq:DualCLP}
\arraycolsep=2pt\def\arraystretch{1.8}
\begin{array}{ll}
\multicolumn{2}{c}{\text{ {\tmstrong{Entanglement quantification}}}} \\[-0.8em]
\multicolumn{2}{c}{\text{\tmstrong{conic linear program} (dual)}} \\
 \text{maximize } & \sum_{a b} \sum_{(x, y) \in \mathcal{I}} \beta_{a b x
 y} P (a b | x y)\\
 \text{over } & \beta_{a b x y} \in \mathbbm{R} \quad \forall a b x y\\
  \text{subject to } &
                       \left (\frac{1}{\dX \dY}, -\sum_{ab}\sum_{(x,y)\in\mathcal{I}}
                       \beta_{abxy} (\xi_x \otimes \psi_y) \right) \in \hat{\mathcal{E}}^*.
 \end{array}
\end{equation}
We observe the following:
\begin{enumerateroman}
 \item The dual objective involves only the dual variables $\beta_{a b x y}$ corresponding to the primal constraint $\tmop{tr} [\Pi_{a b} \cdot (\xi_x \otimes \psi_y)] = P (a b | x y)$.
 \item The dual constraints depend on the definition of the semiquantum scenario (including the input sets $\{ \xi_x \}$, $\{ \psi_y \}$) but not on the coefficients $P (a b | x y)$. 
 \item The coefficients $P (a b | x y)$ are only present in the objective of
 the dual program.
\end{enumerateroman}
Now let $\beta^{\star}_{a b x y}$ be any feasible solution of the dual program~(\ref{Eq:DualCLP}), not necessarily optimal.
We claim it is a quantitative MDIEW. To prove this claim, observe that $\beta_{a b x y}^{\star}$ is a feasible solution of~(\ref{Eq:DualCLP}) regardless of the observations $P (a b | x y)$ --- for a given scenario and entanglement measure.
As a feasible solution, $I^{\star} (\vec{P}) = \sum_{a b x y} \beta_{a b x y}^{\star} P (a b | x y)$ is a lower bound over the maximum value of the dual program.
By duality (see~(\ref{Eq:WeakDuality}) in Appendix~\ref{App:CLP}), $I^{\star} (\vec{P})$ is a lower bound over the optimum of the primal program and thus a valid lower bound on $\mathcal{E}
\left( \rhoAB \right)$.
In particular, this proves our earlier claim made in~\cite{Verbanis2016}, see Appendix~\ref{App:RelationPreviousWork} for details.

\section{Experimental implementations}
\label{Sec:Experimental}

The procedure described in Section~\ref{Sec:Quantifying} works perfectly when the observations $P (a b | x y)$ correspond exactly to the ideal quantum mechanical description of Eq.~(\ref{Eq:PQ}).
In this Section, we deal with the following practical issues.
First, the observations $P (a b | x y)$ are usually frequencies estimated from a finite number of samples, and are thus affected by statistical noise.
Second, while MDIEWs are robust against losses, these losses affect the interpretation of $\nu^{\ast}$ as an entanglement measure.
Finally, we recall that our MDIEWs are robust against classical communication between the devices, so that the measurements do not need to be performed in a spacelike separate manner.

\subsection{Dealing with noisy data}

The observed frequencies $P_{\text{obs}} (a b | x y)$ only approximate the true quantum distribution $P_{\text{Q}} (a b | x y)$, and thus two problems can appear when recovering the effective POVM elements $\{ \Pi_{a b} \}$ from Eq.~(\ref{Eq:Tomography}).
The first problem appears when $\{ \Pi_{a b} \}$ is not a proper POVM because one of its elements $\Pi_{a b}$ has a negative eigenvalue; the corresponding $\mu_{a b}$ is not positive semidefinite, which makes our program~(\ref{Eq:ClaimCLP}) infeasible.
The second problem appears when the set of observed inputs has a linear dependency of the form:
\begin{equation}
 \sum_{(x, y) \in \mathcal{I}} \alpha_{x y} (\xi_x \otimes \psi_y) = 0
\end{equation}
for nonzero coefficients $\alpha_{x y}$ (this happens, for example, when $\{ \xi_x \}$, $\{ \psi_y \}$ contain the six qubit states $| \pm x \rangle, | \pm y \rangle, | \pm z \rangle$).
Then, the linear system~(\ref{Eq:Tomography}) is consistent only when $\sum_{(x, y) \in \mathcal{I}} \alpha_{x y} P (a b | x y) = 0$ --- which, of course, is never the case when dealing with noisy data.

Our conic programs not only quantify the entanglement, but also return a quantitative MDIEW, which enables linear evaluation of the desired entanglement measure.
We thus recommend a procedure in two steps.
In the first step, we work either from simulated data or a first batch of observations $P_{\text{test}} (a b | x y)$.
We find the closest consistent probability distribution $P_{\text{reg}} (a b | x y)$ using a distance such as the Kullback-Leibler divergence or the Euclidean norm.
For example, the Euclidean norm minimization corresponds to the program:
\begin{equation}  \label{Eq:RegularizationSDP} 
\arraycolsep=2pt\def\arraystretch{1.8}
 \begin{array}{ll}
\multicolumn{2}{c}{ \text{{\tmstrong{Regularization semidefinite program}}}} \\
 \text{minimize } & \left\| \vec{P}_{\text{reg}} - \vec{P}_{\text{test}} \right\|_2\\
 \text{over } & \vec{P}_{\text{reg}} \in \mathbbm{R}^{n}_+\\
 & \Pi_{a b} \in \mathsf{Herm}_+ \left( \HX \otimes \HY \right)\\
 \text{such that } & \tmop{tr} [\Pi_{a b} (\xi_x \otimes \psi_y)] =
 P_{\text{reg}} (a b | x y),
 \end{array}
\end{equation}
which can be formulated as a semidefinite program~{\cite{Boyd2004}}; see also~{\cite{Lin2017}} for a formulation of KL-divergence minimization using convex cones.
We use the regularized distribution $P_{\text{reg}} (a b | x y)$ in the program~(\ref{Eq:ClaimCLP}), from which we obtain an estimation of the lower bound on entanglement $\nu^\ast$ and a quantitative MDIEW $\beta_{a b x y}$.
As the estimate $\nu^\ast$ is obtained using regularized data, it can exhibit bias~\cite{Schwemmer2015} or lead to overestimation of entanglement.
In other words, the quantitative MDIEW $\beta_{abxy}$ possibly fits the statistical fluctuations on this first set of data.
The problem is easily avoided by evaluating this MDIEW on a second distribution $P_{\text{exp}} (a b | x y)$ coming from a fresh set of observations without applying any regularization: this corresponds to parameter estimation by linear evaluation, as advocated in~\cite{Schwemmer2015}.

\subsection{Robustness against losses}

The certification of entanglement using MDIEWs is robust against losses or events when no state is produced: we simply consider ``no particle detected'' as one of the possible measurement outcomes (here, we use convention $a = 0$ or $b = 0$ for that outcome), and thus false positives never occur.

\subsubsection{Isotropic losses}

MDIEWs exhibit a stronger property when the losses do not depend on the inputs $(x,y)$.
In that case, even arbitrary high losses do not introduce false negatives.
More precisely, we consider correlations with isotropic losses, which have the form
\begin{equation}
  \label{Eq:Loss}
  P_{\gamma} (a b | x y) = \gamma P_{\text{ideal}} (a b | x y) + (1 - \gamma) P^\gamma_{\emptyset} (a b | x y),
\end{equation}
where $\vec{P}_{\text{ideal}}$ always registers a detection ($\sum_{ab\ge 1}P_{\text{ideal}} (ab | x y) = 1$), while $\vec{P}^\gamma_{\emptyset}$ (whose form can depend on $\gamma$) always registers a nondetection in at least one outcome ${a,b}$:
\begin{equation}
  P^\gamma_\emptyset(00|xy) + \sum_{a\ge 1} P^\gamma_\emptyset (a0|xy) + \sum_{b\ge 1} P^\gamma_\emptyset(0b|xy) = 1 \;.
\end{equation}

Such distributions can model nondetections due to probabilistic state preparation, losses in transmission and detection inefficiencies, provided those are not correlated with the measurement basis. 

If any entanglement can be detected using our method when $\gamma = 1$, then $\vec{P}_{\text{ideal}}$ corresponds to an effective POVM with some of the recovered states $\mu_{\text{ideal}, a b}$ entangled for $a, b \geqslant 1$; and $\nu^{\ast}_{\text{ideal}} = \sum_{a b} p_{\text{ideal}, a b} \mathcal{E} \left( \mu_{\text{ideal}, a b} \right) > 0$.
When $\gamma \rightarrow 0$, the corresponding recovered states are still $\mu_{\gamma, a b} = \mu_{\text{ideal}, a b}$ by linearity of~(\ref{Eq:Loss}) and~(\ref{Eq:Tomography}), but the coefficient $p_{\gamma, a b} = \gamma p_{\text{ideal}, a b}$ is rescaled accordingly.
Thus, we obtain:
\begin{equation}
 \nu^{\ast}_{\gamma} \geqslant \sum_{a b \geqslant 1} \gamma ~ p_{\text{ideal}, a b} ~ \mathcal{E} \left( \mu_{\text{ideal}, a b} \right) = \gamma ~ \nu^{\ast}_{\text{ideal}},
\end{equation}
which is positive for any $\gamma > 0$.
Thus, we can always certify the presence of entanglement when dealing with losses of the form~(\ref{Eq:Loss}).
However, $\nu^{\ast}_{\gamma}$ corresponds to the average entanglement that can be recovered from the semiquantum devices, it is proportional to the probability of a conclusive event ($a, b \geqslant 1$).
This is not surprising: we observe a similar phenomenon when estimating the entanglement in a state with a high ``vacuum'' component, such as:
\begin{equation}
 \rho_{\gamma} = (1 - \gamma) | 00 \rangle \langle 00 | + \gamma \rho_1 \;,
\end{equation}
where $\rho_1 = | 11 + 22 \rangle \langle 11 + 22 |/2$, whose entanglement has upper bound  proportional to $\gamma$ by convexity: $\mathcal{E}(\rho_\gamma) \le \gamma \mathcal{E}(\rho_1)$.

\subsubsection{Quantifying entanglement after local filtering}

In some applications, we want to quantify entanglement after discarding inconclusive events, for example in an experiment involving continuous-wave-pumped sources where events are only recorded when successful detection occurs.
If we postselect on conclusive events, the reported figure of merit needs to be interpreted carefully~{\cite{Lim2016}}: while entanglement measures are monotonic under LOCC operations, they are not monotonic after postselection (SLOCC).
After such postselection, the quantity we estimate corresponds to the amount of entanglement $\mathcal{E}^{\text{SLOCC}} \left( \rhoAB \right)$ after possible local filtering~{\cite{Gisin1996,Verstraete2001a}}:
\begin{multline}
 \mathcal{E}^{\text{SLOCC}} \left( \rhoAB \right) = \max_{A, B} \mathcal{E} \left (\frac{(A \otimes B) \rhoAB (A \otimes B)^{\dag}}{\tmop{tr} \left[ (A \otimes B) \rhoAB (A \otimes B)^{\dag} \right]} \right ) \;, \\
 \mathbbm{1} - A^{\dag} A \succcurlyeq 0, \quad \mathbbm{1} - B^{\dag} B \succcurlyeq 0,
\end{multline}
and $\mathcal{E}^{\text{SLOCC}} \left( \rhoAB \right) \geqslant \mathcal{E} \left( \rhoAB \right)$.
Still, filtering/postselection can only increase the amount of entanglement within some limitations.
For example, the negativity of qubit Werner states cannot be increased by SLOCC~{\cite{Verstraete2001,Verstraete2002}}.

By construction, the extraction of the recovered state $\mu_{a b}$, postselected on $(a, b)$, is a state that can be obtained from $\rhoAB$ by SLOCC (\figurename{~\ref{Fig:EffectivePOVM}b}).
Thus:
\begin{equation}
 \mathcal{E}^{\text{SLOCC}} \left( \rhoAB \right) \geqslant \max_{a b} \mathcal{E} (\mu_{a b}),
\end{equation}
and we directly obtain a lower bound on $\mathcal{E}^{\text{SLOCC}}$ from the $\mu_{a b}$.

\subsubsection{When the total number of events is unknown}
In our previous experimental work~{\cite{Verbanis2016}}, we obtained the number of conclusive events $N (a b x y)$ for the outcomes $a, b \geqslant 1$, while the number of inconclusive events $N (\emptyset x y)$ was unknown.
In these inconclusive events, we collect under the notation ``$\emptyset$'' all outcomes with at least one of $a, b = 0$.

Now, we cannot simply ignore the inconclusive events as is usually done under the assumption of fair-sampling:

\begin{equation}
  P_\text{FS} (a b | x y) = \frac{N (a b x y)}{\sum_{a b} N (a b x y)},
\end{equation}
because we cannot assume that the number of conclusive events is the same regardless of the input pair, as we make no assumptions about the measurement devices.
However, when the experiment is run for all input pairs $(x, y)$ with constant efficiency and duration, we have $N (x y) = N (\emptyset x y) + \sum_{a b} N (a b x y) = N^{\ast}$ constant.

The true value of $N^{\ast}$ is unknown, however it satisfies $N^{\ast} \geqslant \sum_{a b} N (a b x y)$.
We can set it to its lower bound:
\begin{equation}
 N^{\ast} = \max_{x y} \sum_{a b} N (a b x y),
\end{equation}
so that:
\begin{equation}
 \label{Eq:DistributionFromCounts} \tilde{P} (a b | x y) = \frac{N (a b x y)}{N^{\ast}}, \, \tilde{P} (\emptyset | x y) = 1 - \sum_{a b} \tilde{P} (a b | x y)
\end{equation}
represents a guess that only differs from the true distribution $P (a b | x y)$ by a constant factor (for $ab \ge 1$):
\begin{equation}
 P (a b | x y) = \alpha \tilde{P} (a b | x y), \qquad \alpha \in] 0, 1] .
\end{equation}
The effective POVM $\{ \tilde{\Pi}_{a b} \}$ differs from $\{ \Pi_{a b} \}$ by the same constant factor, thus our guess $\tilde{P} (a b | x y)$ leads to the same recovered states $\tilde{\mu}_{a b}$ as the true $P (a b | x y)$ for $a, b \geqslant 1$.
We now write:
\begin{equation}
\begin{split}
 \overline{\nu} = \sum_{a b \geqslant 1} \tilde{p}_{a b} \mathcal{E} (\tilde{\mu}_{a b}) &= \frac{\sum_{a b \geqslant 1} \tilde{p}_{a b} \mathcal{E} (\tilde{\mu}_{a b})}{\tilde{p}_{\emptyset} + \sum_{a b \geqslant 1} \tilde{p}_{a b}} \\
& \leqslant \frac{\sum_{a b \geqslant 1} \tilde{p}_{a b} \mathcal{E} (\tilde{\mu}_{a b})}{\sum_{a b \geqslant 1} \tilde{p}_{a b}} \\
&\leqslant \max_{a b} \mathcal{E} (\tilde{\mu}_{a b}) = \max_{a b} \mathcal{E} (\mu_{a b}),
\end{split}
\end{equation}
and observe that $\overline{\nu}$ can be computed using our conic linear program~(\ref{Eq:ClaimCLPDuplicate}) with the subnormalized data $\tilde{P} (a b | x y)$ for $a, b \geqslant 1$; then $\overline{\nu}$ is a lower bound on $\mathcal{E}^{\text{SLOCC}} \left( \rhoAB \right)$.

\subsection{Robustness against classical communication}

In Section~\ref{Sec:EntanglementDistributed}, we defined our bound $\nu$ with respect to the entanglement production capacity of the effective POVM $\{ \Pi_{a b} \}$, which in turn provides a lower bound on the entanglement of the shared state $\rhoAB$.
These bounds hold because entanglement measures are monotonic under local operations, also when allowing local communication (LOCC).
Thus, our quantitative MDIEWs are robust against classical communication.
This eases the experimental requirements as the measurements do not need to be performed in a spacelike separated manner.

What about {\textit{quantum}} communication between the devices? Indeed, the device A could transmit its quantum input to the device B, and then the joint measurement can be performed in a single device with the outcome $a$ transmitted back to A.
In that case however, the violation of our MDIEW shows the existence of a quantum channel that preserves entanglement; and this channel can also be used to establish an entangled state between A and B (we will consider the characterization of quantum channels using measurement-device-independent witnesses in another work~{\cite{Rosset2017}}).

\subsection{Experimental application}

We now come back to the experimental observations of our previous Letter~{\cite{Verbanis2016}}, where we obtained data sets for a Werner state:
\begin{equation}
 \label{Eq:WernerState} \rhoAB = \lambda \ketbra{\varphi_2} + (1 - \lambda) \mathbbm{1} / 4,
\end{equation}
for $\lambda$ between $0.29$ and $0.94$. Each data set contains the event counts $N (a b x y)$ for conclusive events, from which we estimate the frequencies $\tilde{P} (a b | x y)$ according to the procedure described in Eq.~(\ref{Eq:DistributionFromCounts}).
We construct a quantitative MDIEW for the following entanglement measures: the negativity, the absolute, generalized and random robustness, and the semidefinite upper bound on entanglement distillation.
For each entanglement measure, we take the frequencies $\tilde{P} (a b | x y)$ for the Bell state fraction $\lambda = 0.94$ and compute the regularized distribution $P_{\text{reg}} (a b | x y)$ using~(\ref{Eq:RegularizationSDP}). We finally run the conic program using CVX~{\cite{Grant2014}} and Mosek~{\cite{ApS2015}} on $\vec{P}_{\text{reg}}$ to obtain $\nu^{\ast}$ and the dual variables $\beta_{a b x y}$ corresponding to a quantitative MDIEW.
We then evaluate $\sum_{a b x y} \beta_{a b x y} \tilde{P} (a b | x y) = \vec{\beta} \cdot \tilde{P}$ for all the data sets without applying regularization.
We present the results in \figurename{~\ref{Fig:ExperimentalResults}}, illustrating the versatility of our method.

\begin{figure}[h]
\includegraphics[width=\columnwidth]{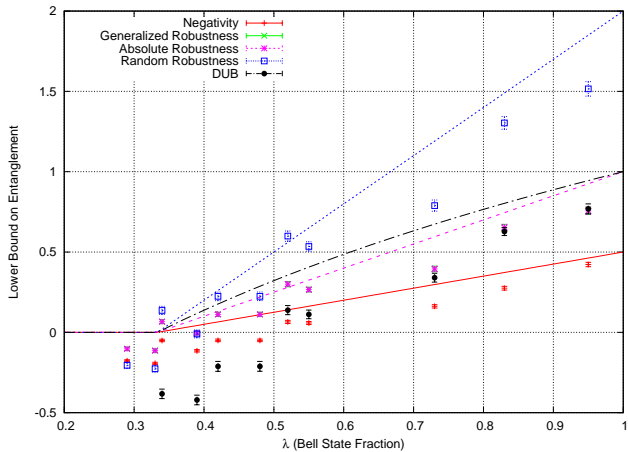}
\caption{
  \label{Fig:ExperimentalResults}
  Construction and evaluation of several entanglement measures on the experimental observations described in~{\cite{Verbanis2016}}.
  The points correspond to the values obtained by the evaluation of the MDIEW $\vec{\beta} \cdot \tilde{P}$ directly on the estimated frequencies, while the continuous line corresponds to the entanglement present in the ideal Werner state~(\ref{Eq:WernerState}).
  DUB is the upper bound on distillable entanglement due to Wang and Duan~{\cite{Wang2016}}, and all the entanglement measures are described in Appendix~\ref{App:RecipesEntanglement}.
}
\end{figure}

\section{Conclusion}

We proposed a practical method to construct quantitative measurement-device-independent entanglement witnesses purely from observational data, corresponding to lower bounds on well-known convex computable entanglement measures.
In particular, we proved the validity of earlier claims~{\cite{Verbanis2016}} related to the quantification of negativity using a resource efficient MDIEW implementation.

Compared to earlier approaches, our method does not prescribe a particular measurement basis, does not require tomographically complete sets of inputs, or even data for all input pairs.
As with other MDIEW constructions, our witnesses are robust against imperfections in the measurement devices and against classical communication, so that measurements do not need to be space-like separated.

Our numerical implementations correspond to conic linear programs, whose formulations are inspired by the disciplined convex programming approach~{\cite{Grant2006}}, where convex functions are encapsulated in components.
We identified the relevant components for the quantification of entanglement, the convex cones corresponding to entanglement measures, and presented recipes for the evaluation of main convex entanglement measures.
While we applied these formulations in semiquantum scenarios, they easily translate to other scenarios, for example device-independent entanglement quantification~{\cite{Moroder2013}} or entanglement quantification in steering scenarios~{\cite{Chen2016}}.

While we focused on the bipartite case, our results can be extended to the multipartite case, using the corresponding distributed effective POVM $\{ \Pi_{a b c \ldots} \}$ and characterizing the type of entanglement present in the recovered states $\mu_{a b c \ldots}$.
In the future, we would like to complete our library of recipes by adding more computable convex measures, for example by including the semidefinite formulations presented in~{\cite{Toth2015}}.

The construction presented in this paper rests on the extraction of entangled states from an effective POVM with entangled elements; we prove that our construction is optimal when the measurement devices perform a generalized Bell measurement.
Other constructions for MDIEWs likely exist, as hinted at by the existence of POVMs with separable elements that nevertheless require entanglement to be performed~{\cite{Bennett1999}}; we leave this study as an open question.

{\noindent}
\tmtextbf{Acknowledgments. }
We thank Jean-Daniel Bancal, Tomer Barnea, Cyril Branciard, Francesco Buscemi, Nicolas Gisin, Yeong-Cherng Liang, Martin Ringbauer and Hugo Zbinden for discussions.
Research at Perimeter Institute is supported by the Government of Canada through Industry Canada and by the Province of Ontario through the Ministry of Research and Innovation.
This publication was made possible through the support of a grant from the John Templeton Foundation; additionally, this work was supported by the Swiss National Science Foundation (Grant No. 200021\_159592), D.R. was supported by the SNSF Early Postdoc. Mobility fellowship P2GEP2\_162060, and C.C.W.L. acknowledges support from NUS start-up grant R-263-000-C78-133/731.

\paragraph*{Note added}
While completing our manuscript, we became aware of related work~{\cite{Supic2017a}} also considering the quantification of entanglement in measurement-device-independent scenarios.
We discuss the differences between our approaches in Appendix~\ref{App:Relation}.

\bibliographystyle{apsrev4-1} 
\bibliography{LongMdiew}

\begin{thebibliography}{67}%
\makeatletter
\providecommand \@ifxundefined [1]{%
 \@ifx{#1\undefined}
}%
\providecommand \@ifnum [1]{%
 \ifnum #1\expandafter \@firstoftwo
 \else \expandafter \@secondoftwo
 \fi
}%
\providecommand \@ifx [1]{%
 \ifx #1\expandafter \@firstoftwo
 \else \expandafter \@secondoftwo
 \fi
}%
\providecommand \natexlab [1]{#1}%
\providecommand \enquote  [1]{``#1''}%
\providecommand \bibnamefont  [1]{#1}%
\providecommand \bibfnamefont [1]{#1}%
\providecommand \citenamefont [1]{#1}%
\providecommand \href@noop [0]{\@secondoftwo}%
\providecommand \href [0]{\begingroup \@sanitize@url \@href}%
\providecommand \@href[1]{\@@startlink{#1}\@@href}%
\providecommand \@@href[1]{\endgroup#1\@@endlink}%
\providecommand \@sanitize@url [0]{\catcode `\\12\catcode `\$12\catcode
  `\&12\catcode `\#12\catcode `\^12\catcode `\_12\catcode `\%12\relax}%
\providecommand \@@startlink[1]{}%
\providecommand \@@endlink[0]{}%
\providecommand \url  [0]{\begingroup\@sanitize@url \@url }%
\providecommand \@url [1]{\endgroup\@href {#1}{\urlprefix }}%
\providecommand \urlprefix  [0]{URL }%
\providecommand \Eprint [0]{\href }%
\providecommand \doibase [0]{http://dx.doi.org/}%
\providecommand \selectlanguage [0]{\@gobble}%
\providecommand \bibinfo  [0]{\@secondoftwo}%
\providecommand \bibfield  [0]{\@secondoftwo}%
\providecommand \translation [1]{[#1]}%
\providecommand \BibitemOpen [0]{}%
\providecommand \bibitemStop [0]{}%
\providecommand \bibitemNoStop [0]{.\EOS\space}%
\providecommand \EOS [0]{\spacefactor3000\relax}%
\providecommand \BibitemShut  [1]{\csname bibitem#1\endcsname}%
\let\auto@bib@innerbib\@empty
\bibitem [{\citenamefont {Horodecki}\ \emph {et~al.}(2009)\citenamefont
  {Horodecki}, \citenamefont {Horodecki}, \citenamefont {Horodecki},\ and\
  \citenamefont {Horodecki}}]{Horodecki2009}%
  \BibitemOpen
  \bibfield  {author} {\bibinfo {author} {\bibfnamefont {R.}~\bibnamefont
  {Horodecki}}, \bibinfo {author} {\bibfnamefont {P.}~\bibnamefont
  {Horodecki}}, \bibinfo {author} {\bibfnamefont {M.}~\bibnamefont
  {Horodecki}}, \ and\ \bibinfo {author} {\bibfnamefont {K.}~\bibnamefont
  {Horodecki}},\ }\href {\doibase 10.1103/RevModPhys.81.865} {\bibfield
  {journal} {\bibinfo  {journal} {Rev. Mod. Phys.}\ }\textbf {\bibinfo {volume}
  {81}},\ \bibinfo {pages} {865} (\bibinfo {year} {2009})}\BibitemShut
  {NoStop}%
\bibitem [{\citenamefont {Ekert}(1991)}]{Ekert1991}%
  \BibitemOpen
  \bibfield  {author} {\bibinfo {author} {\bibfnamefont {A.~K.}\ \bibnamefont
  {Ekert}},\ }\href {\doibase 10.1103/PhysRevLett.67.661} {\bibfield  {journal}
  {\bibinfo  {journal} {Phys. Rev. Lett.}\ }\textbf {\bibinfo {volume} {67}},\
  \bibinfo {pages} {661} (\bibinfo {year} {1991})}\BibitemShut {NoStop}%
\bibitem [{\citenamefont {Jozsa}\ and\ \citenamefont
  {Linden}(2003)}]{Jozsa2003}%
  \BibitemOpen
  \bibfield  {author} {\bibinfo {author} {\bibfnamefont {R.}~\bibnamefont
  {Jozsa}}\ and\ \bibinfo {author} {\bibfnamefont {N.}~\bibnamefont {Linden}},\
  }\href {\doibase 10.1098/rspa.2002.1097} {\bibfield  {journal} {\bibinfo
  {journal} {Proc. R. Soc. A Math. Phys. Eng. Sci.}\ }\textbf {\bibinfo
  {volume} {459}},\ \bibinfo {pages} {2011} (\bibinfo {year}
  {2003})}\BibitemShut {NoStop}%
\bibitem [{\citenamefont {Horodecki}(2001)}]{Horodecki2001}%
  \BibitemOpen
  \bibfield  {author} {\bibinfo {author} {\bibfnamefont {M.}~\bibnamefont
  {Horodecki}},\ }\href@noop {} {\bibfield  {journal} {\bibinfo  {journal}
  {Quantum Info. Comput.}\ }\textbf {\bibinfo {volume} {1}},\ \bibinfo {pages}
  {3} (\bibinfo {year} {2001})}\BibitemShut {NoStop}%
\bibitem [{\citenamefont {Rains}(1999)}]{Rains1999a}%
  \BibitemOpen
  \bibfield  {author} {\bibinfo {author} {\bibfnamefont {E.~M.}\ \bibnamefont
  {Rains}},\ }\href {\doibase 10.1103/PhysRevA.60.173} {\bibfield  {journal}
  {\bibinfo  {journal} {Phys. Rev. A}\ }\textbf {\bibinfo {volume} {60}},\
  \bibinfo {pages} {173} (\bibinfo {year} {1999})}\BibitemShut {NoStop}%
\bibitem [{\citenamefont {Hayden}\ \emph {et~al.}(2001)\citenamefont {Hayden},
  \citenamefont {Horodecki},\ and\ \citenamefont {Terhal}}]{Hayden2001}%
  \BibitemOpen
  \bibfield  {author} {\bibinfo {author} {\bibfnamefont {P.~M.}\ \bibnamefont
  {Hayden}}, \bibinfo {author} {\bibfnamefont {M.}~\bibnamefont {Horodecki}}, \
  and\ \bibinfo {author} {\bibfnamefont {B.~M.}\ \bibnamefont {Terhal}},\
  }\href {\doibase 10.1088/0305-4470/34/35/314} {\bibfield  {journal} {\bibinfo
   {journal} {J. Phys. A. Math. Gen.}\ }\textbf {\bibinfo {volume} {34}},\
  \bibinfo {pages} {6891} (\bibinfo {year} {2001})}\BibitemShut {NoStop}%
\bibitem [{\citenamefont {Wootters}(1998)}]{Wootters1998}%
  \BibitemOpen
  \bibfield  {author} {\bibinfo {author} {\bibfnamefont {W.~K.}\ \bibnamefont
  {Wootters}},\ }\href {\doibase 10.1103/PhysRevLett.80.2245} {\bibfield
  {journal} {\bibinfo  {journal} {Phys. Rev. Lett.}\ }\textbf {\bibinfo
  {volume} {80}},\ \bibinfo {pages} {2245} (\bibinfo {year}
  {1998})}\BibitemShut {NoStop}%
\bibitem [{\citenamefont {Vedral}(2002)}]{Vedral2002}%
  \BibitemOpen
  \bibfield  {author} {\bibinfo {author} {\bibfnamefont {V.}~\bibnamefont
  {Vedral}},\ }\href {\doibase 10.1103/RevModPhys.74.197} {\bibfield  {journal}
  {\bibinfo  {journal} {Rev. Mod. Phys.}\ }\textbf {\bibinfo {volume} {74}},\
  \bibinfo {pages} {197} (\bibinfo {year} {2002})}\BibitemShut {NoStop}%
\bibitem [{\citenamefont {Vidal}\ and\ \citenamefont
  {Tarrach}(1999)}]{Vidal1999}%
  \BibitemOpen
  \bibfield  {author} {\bibinfo {author} {\bibfnamefont {G.}~\bibnamefont
  {Vidal}}\ and\ \bibinfo {author} {\bibfnamefont {R.}~\bibnamefont
  {Tarrach}},\ }\href {\doibase 10.1103/PhysRevA.59.141} {\bibfield  {journal}
  {\bibinfo  {journal} {Phys. Rev. A}\ }\textbf {\bibinfo {volume} {59}},\
  \bibinfo {pages} {141} (\bibinfo {year} {1999})}\BibitemShut {NoStop}%
\bibitem [{\citenamefont {Steiner}(2003)}]{Steiner2003}%
  \BibitemOpen
  \bibfield  {author} {\bibinfo {author} {\bibfnamefont {M.}~\bibnamefont
  {Steiner}},\ }\href {\doibase 10.1103/PhysRevA.67.054305} {\bibfield
  {journal} {\bibinfo  {journal} {Phys. Rev. A}\ }\textbf {\bibinfo {volume}
  {67}},\ \bibinfo {pages} {054305} (\bibinfo {year} {2003})}\BibitemShut
  {NoStop}%
\bibitem [{\citenamefont {Vidal}\ and\ \citenamefont
  {Werner}(2002)}]{Vidal2002}%
  \BibitemOpen
  \bibfield  {author} {\bibinfo {author} {\bibfnamefont {G.}~\bibnamefont
  {Vidal}}\ and\ \bibinfo {author} {\bibfnamefont {R.~F.}\ \bibnamefont
  {Werner}},\ }\href {\doibase 10.1103/PhysRevA.65.032314} {\bibfield
  {journal} {\bibinfo  {journal} {Phys. Rev. A}\ }\textbf {\bibinfo {volume}
  {65}},\ \bibinfo {pages} {032314} (\bibinfo {year} {2002})}\BibitemShut
  {NoStop}%
\bibitem [{\citenamefont {Wang}\ and\ \citenamefont {Duan}(2016)}]{Wang2016}%
  \BibitemOpen
  \bibfield  {author} {\bibinfo {author} {\bibfnamefont {X.}~\bibnamefont
  {Wang}}\ and\ \bibinfo {author} {\bibfnamefont {R.}~\bibnamefont {Duan}},\
  }\href {\doibase 10.1103/PhysRevA.94.050301} {\bibfield  {journal} {\bibinfo
  {journal} {Phys. Rev. A}\ }\textbf {\bibinfo {volume} {94}},\ \bibinfo
  {pages} {050301} (\bibinfo {year} {2016})}\BibitemShut {NoStop}%
\bibitem [{\citenamefont {Plenio}\ and\ \citenamefont
  {Virmani}(2014)}]{Plenio2007}%
  \BibitemOpen
  \bibfield  {author} {\bibinfo {author} {\bibfnamefont {M.~B.}\ \bibnamefont
  {Plenio}}\ and\ \bibinfo {author} {\bibfnamefont {S.~S.}\ \bibnamefont
  {Virmani}},\ }in\ \href {\doibase 10.1007/978-3-319-04063-9_8} {\emph
  {\bibinfo {booktitle} {Quantum Inf. Coherence}}},\ Vol.~\bibinfo {volume}
  {7}\ (\bibinfo {year} {2014})\ pp.\ \bibinfo {pages} {173--209}\BibitemShut
  {NoStop}%
\bibitem [{\citenamefont {Eltschka}\ and\ \citenamefont
  {Siewert}(2014)}]{Eltschka2014}%
  \BibitemOpen
  \bibfield  {author} {\bibinfo {author} {\bibfnamefont {C.}~\bibnamefont
  {Eltschka}}\ and\ \bibinfo {author} {\bibfnamefont {J.}~\bibnamefont
  {Siewert}},\ }\href {\doibase 10.1088/1751-8113/47/42/424005} {\bibfield
  {journal} {\bibinfo  {journal} {J. Phys. A Math. Theor.}\ }\textbf {\bibinfo
  {volume} {47}},\ \bibinfo {pages} {424005} (\bibinfo {year}
  {2014})}\BibitemShut {NoStop}%
\bibitem [{\citenamefont {Huang}(2014)}]{Huang2014}%
  \BibitemOpen
  \bibfield  {author} {\bibinfo {author} {\bibfnamefont {Y.}~\bibnamefont
  {Huang}},\ }\href {\doibase 10.1088/1367-2630/16/3/033027} {\bibfield
  {journal} {\bibinfo  {journal} {New J. Phys.}\ }\textbf {\bibinfo {volume}
  {16}},\ \bibinfo {pages} {033027} (\bibinfo {year} {2014})}\BibitemShut
  {NoStop}%
\bibitem [{\citenamefont {Schwemmer}\ \emph {et~al.}(2015)\citenamefont
  {Schwemmer}, \citenamefont {Knips}, \citenamefont {Richart}, \citenamefont
  {Weinfurter}, \citenamefont {Moroder}, \citenamefont {Kleinmann},\ and\
  \citenamefont {G{\"{u}}hne}}]{Schwemmer2015}%
  \BibitemOpen
  \bibfield  {author} {\bibinfo {author} {\bibfnamefont {C.}~\bibnamefont
  {Schwemmer}}, \bibinfo {author} {\bibfnamefont {L.}~\bibnamefont {Knips}},
  \bibinfo {author} {\bibfnamefont {D.}~\bibnamefont {Richart}}, \bibinfo
  {author} {\bibfnamefont {H.}~\bibnamefont {Weinfurter}}, \bibinfo {author}
  {\bibfnamefont {T.}~\bibnamefont {Moroder}}, \bibinfo {author} {\bibfnamefont
  {M.}~\bibnamefont {Kleinmann}}, \ and\ \bibinfo {author} {\bibfnamefont
  {O.}~\bibnamefont {G{\"{u}}hne}},\ }\href {\doibase
  10.1103/PhysRevLett.114.080403} {\bibfield  {journal} {\bibinfo  {journal}
  {Phys. Rev. Lett.}\ }\textbf {\bibinfo {volume} {114}},\ \bibinfo {pages}
  {080403} (\bibinfo {year} {2015})}\BibitemShut {NoStop}%
\bibitem [{\citenamefont {Rosset}\ \emph {et~al.}(2012)\citenamefont {Rosset},
  \citenamefont {Ferretti-Sch{\"{o}}bitz}, \citenamefont {Bancal},
  \citenamefont {Gisin},\ and\ \citenamefont {Liang}}]{Rosset2012a}%
  \BibitemOpen
  \bibfield  {author} {\bibinfo {author} {\bibfnamefont {D.}~\bibnamefont
  {Rosset}}, \bibinfo {author} {\bibfnamefont {R.}~\bibnamefont
  {Ferretti-Sch{\"{o}}bitz}}, \bibinfo {author} {\bibfnamefont {J.-D.}\
  \bibnamefont {Bancal}}, \bibinfo {author} {\bibfnamefont {N.}~\bibnamefont
  {Gisin}}, \ and\ \bibinfo {author} {\bibfnamefont {Y.-C.}\ \bibnamefont
  {Liang}},\ }\href {\doibase 10.1103/PhysRevA.86.062325} {\bibfield  {journal}
  {\bibinfo  {journal} {Phys. Rev. A}\ }\textbf {\bibinfo {volume} {86}},\
  \bibinfo {pages} {062325} (\bibinfo {year} {2012})}\BibitemShut {NoStop}%
\bibitem [{\citenamefont {Eisert}\ \emph {et~al.}(2007)\citenamefont {Eisert},
  \citenamefont {Brand{\~{a}}o},\ and\ \citenamefont {Audenaert}}]{Eisert2007}%
  \BibitemOpen
  \bibfield  {author} {\bibinfo {author} {\bibfnamefont {J.}~\bibnamefont
  {Eisert}}, \bibinfo {author} {\bibfnamefont {F.~G. S.~L.}\ \bibnamefont
  {Brand{\~{a}}o}}, \ and\ \bibinfo {author} {\bibfnamefont {K.~M.~R.}\
  \bibnamefont {Audenaert}},\ }\href {\doibase 10.1088/1367-2630/9/3/046}
  {\bibfield  {journal} {\bibinfo  {journal} {New J. Phys.}\ }\textbf {\bibinfo
  {volume} {9}},\ \bibinfo {pages} {46} (\bibinfo {year} {2007})}\BibitemShut
  {NoStop}%
\bibitem [{\citenamefont {Bell}(1964)}]{Bell1964}%
  \BibitemOpen
  \bibfield  {author} {\bibinfo {author} {\bibfnamefont {J.~S.}\ \bibnamefont
  {Bell}},\ }\href@noop {} {\bibfield  {journal} {\bibinfo  {journal} {Phys.
  (Long Isl. City, N.Y.)}\ }\textbf {\bibinfo {volume} {1}},\ \bibinfo {pages}
  {195} (\bibinfo {year} {1964})}\BibitemShut {NoStop}%
\bibitem [{\citenamefont {Brunner}\ \emph {et~al.}(2014)\citenamefont
  {Brunner}, \citenamefont {Cavalcanti}, \citenamefont {Pironio}, \citenamefont
  {Scarani},\ and\ \citenamefont {Wehner}}]{Brunner2014}%
  \BibitemOpen
  \bibfield  {author} {\bibinfo {author} {\bibfnamefont {N.}~\bibnamefont
  {Brunner}}, \bibinfo {author} {\bibfnamefont {D.}~\bibnamefont {Cavalcanti}},
  \bibinfo {author} {\bibfnamefont {S.}~\bibnamefont {Pironio}}, \bibinfo
  {author} {\bibfnamefont {V.}~\bibnamefont {Scarani}}, \ and\ \bibinfo
  {author} {\bibfnamefont {S.}~\bibnamefont {Wehner}},\ }\href {\doibase
  10.1103/RevModPhys.86.419} {\bibfield  {journal} {\bibinfo  {journal} {Rev.
  Mod. Phys.}\ }\textbf {\bibinfo {volume} {86}},\ \bibinfo {pages} {419}
  (\bibinfo {year} {2014})}\BibitemShut {NoStop}%
\bibitem [{\citenamefont {Moroder}\ \emph {et~al.}(2013)\citenamefont
  {Moroder}, \citenamefont {Bancal}, \citenamefont {Liang}, \citenamefont
  {Hofmann},\ and\ \citenamefont {G{\"{u}}hne}}]{Moroder2013}%
  \BibitemOpen
  \bibfield  {author} {\bibinfo {author} {\bibfnamefont {T.}~\bibnamefont
  {Moroder}}, \bibinfo {author} {\bibfnamefont {J.-D.}\ \bibnamefont {Bancal}},
  \bibinfo {author} {\bibfnamefont {Y.-C.}\ \bibnamefont {Liang}}, \bibinfo
  {author} {\bibfnamefont {M.}~\bibnamefont {Hofmann}}, \ and\ \bibinfo
  {author} {\bibfnamefont {O.}~\bibnamefont {G{\"{u}}hne}},\ }\href {\doibase
  10.1103/PhysRevLett.111.030501} {\bibfield  {journal} {\bibinfo  {journal}
  {Phys. Rev. Lett.}\ }\textbf {\bibinfo {volume} {111}},\ \bibinfo {pages}
  {030501} (\bibinfo {year} {2013})}\BibitemShut {NoStop}%
\bibitem [{\citenamefont {Werner}(1989)}]{Werner1989}%
  \BibitemOpen
  \bibfield  {author} {\bibinfo {author} {\bibfnamefont {R.~F.}\ \bibnamefont
  {Werner}},\ }\href {\doibase 10.1103/PhysRevA.40.4277} {\bibfield  {journal}
  {\bibinfo  {journal} {Phys. Rev. A}\ }\textbf {\bibinfo {volume} {40}},\
  \bibinfo {pages} {4277} (\bibinfo {year} {1989})}\BibitemShut {NoStop}%
\bibitem [{\citenamefont {Degorre}\ \emph {et~al.}(2005)\citenamefont
  {Degorre}, \citenamefont {Laplante},\ and\ \citenamefont
  {Roland}}]{Degorre2005}%
  \BibitemOpen
  \bibfield  {author} {\bibinfo {author} {\bibfnamefont {J.}~\bibnamefont
  {Degorre}}, \bibinfo {author} {\bibfnamefont {S.}~\bibnamefont {Laplante}}, \
  and\ \bibinfo {author} {\bibfnamefont {J.}~\bibnamefont {Roland}},\ }\href
  {\doibase 10.1103/PhysRevA.72.062314} {\bibfield  {journal} {\bibinfo
  {journal} {Phys. Rev. A}\ }\textbf {\bibinfo {volume} {72}},\ \bibinfo
  {pages} {062314} (\bibinfo {year} {2005})}\BibitemShut {NoStop}%
\bibitem [{\citenamefont {Barrett}(2002)}]{Barrett2002}%
  \BibitemOpen
  \bibfield  {author} {\bibinfo {author} {\bibfnamefont {J.}~\bibnamefont
  {Barrett}},\ }\href {\doibase 10.1103/PhysRevA.65.042302} {\bibfield
  {journal} {\bibinfo  {journal} {Phys. Rev. A}\ }\textbf {\bibinfo {volume}
  {65}},\ \bibinfo {pages} {042302} (\bibinfo {year} {2002})}\BibitemShut
  {NoStop}%
\bibitem [{\citenamefont {Branciard}\ \emph {et~al.}(2013)\citenamefont
  {Branciard}, \citenamefont {Rosset}, \citenamefont {Liang},\ and\
  \citenamefont {Gisin}}]{Branciard2013}%
  \BibitemOpen
  \bibfield  {author} {\bibinfo {author} {\bibfnamefont {C.}~\bibnamefont
  {Branciard}}, \bibinfo {author} {\bibfnamefont {D.}~\bibnamefont {Rosset}},
  \bibinfo {author} {\bibfnamefont {Y.-C.}\ \bibnamefont {Liang}}, \ and\
  \bibinfo {author} {\bibfnamefont {N.}~\bibnamefont {Gisin}},\ }\href
  {\doibase 10.1103/PhysRevLett.110.060405} {\bibfield  {journal} {\bibinfo
  {journal} {Phys. Rev. Lett.}\ }\textbf {\bibinfo {volume} {110}},\ \bibinfo
  {pages} {060405} (\bibinfo {year} {2013})}\BibitemShut {NoStop}%
\bibitem [{\citenamefont {Buscemi}(2012)}]{Buscemi2012}%
  \BibitemOpen
  \bibfield  {author} {\bibinfo {author} {\bibfnamefont {F.}~\bibnamefont
  {Buscemi}},\ }\href {\doibase 10.1103/PhysRevLett.108.200401} {\bibfield
  {journal} {\bibinfo  {journal} {Phys. Rev. Lett.}\ }\textbf {\bibinfo
  {volume} {108}},\ \bibinfo {pages} {200401} (\bibinfo {year}
  {2012})}\BibitemShut {NoStop}%
\bibitem [{\citenamefont {Rosset}\ \emph {et~al.}(2013)\citenamefont {Rosset},
  \citenamefont {Branciard}, \citenamefont {Gisin},\ and\ \citenamefont
  {Liang}}]{Rosset2013a}%
  \BibitemOpen
  \bibfield  {author} {\bibinfo {author} {\bibfnamefont {D.}~\bibnamefont
  {Rosset}}, \bibinfo {author} {\bibfnamefont {C.}~\bibnamefont {Branciard}},
  \bibinfo {author} {\bibfnamefont {N.}~\bibnamefont {Gisin}}, \ and\ \bibinfo
  {author} {\bibfnamefont {Y.-C.}\ \bibnamefont {Liang}},\ }\href {\doibase
  10.1088/1367-2630/15/5/053025} {\bibfield  {journal} {\bibinfo  {journal}
  {New J. Phys.}\ }\textbf {\bibinfo {volume} {15}},\ \bibinfo {pages} {053025}
  (\bibinfo {year} {2013})}\BibitemShut {NoStop}%
\bibitem [{\citenamefont {Xu}\ \emph {et~al.}(2014)\citenamefont {Xu},
  \citenamefont {Yuan}, \citenamefont {Chen}, \citenamefont {Lu}, \citenamefont
  {Yao}, \citenamefont {Ma}, \citenamefont {Chen},\ and\ \citenamefont
  {Pan}}]{Xu2014}%
  \BibitemOpen
  \bibfield  {author} {\bibinfo {author} {\bibfnamefont {P.}~\bibnamefont
  {Xu}}, \bibinfo {author} {\bibfnamefont {X.}~\bibnamefont {Yuan}}, \bibinfo
  {author} {\bibfnamefont {L.-K.}\ \bibnamefont {Chen}}, \bibinfo {author}
  {\bibfnamefont {H.}~\bibnamefont {Lu}}, \bibinfo {author} {\bibfnamefont
  {X.~C.}\ \bibnamefont {Yao}}, \bibinfo {author} {\bibfnamefont
  {X.}~\bibnamefont {Ma}}, \bibinfo {author} {\bibfnamefont {Y.~A.}\
  \bibnamefont {Chen}}, \ and\ \bibinfo {author} {\bibfnamefont {J.-W.}\
  \bibnamefont {Pan}},\ }\href {\doibase 10.1103/PhysRevLett.112.140506}
  {\bibfield  {journal} {\bibinfo  {journal} {Phys. Rev. Lett.}\ }\textbf
  {\bibinfo {volume} {112}},\ \bibinfo {pages} {140506} (\bibinfo {year}
  {2014})}\BibitemShut {NoStop}%
\bibitem [{\citenamefont {Nawareg}\ \emph {et~al.}(2015)\citenamefont
  {Nawareg}, \citenamefont {Muhammad}, \citenamefont {Amselem},\ and\
  \citenamefont {Bourennane}}]{Nawareg2015}%
  \BibitemOpen
  \bibfield  {author} {\bibinfo {author} {\bibfnamefont {M.}~\bibnamefont
  {Nawareg}}, \bibinfo {author} {\bibfnamefont {S.}~\bibnamefont {Muhammad}},
  \bibinfo {author} {\bibfnamefont {E.}~\bibnamefont {Amselem}}, \ and\
  \bibinfo {author} {\bibfnamefont {M.}~\bibnamefont {Bourennane}},\ }\href
  {\doibase 10.1038/srep08048} {\bibfield  {journal} {\bibinfo  {journal} {Sci.
  Rep.}\ }\textbf {\bibinfo {volume} {5}},\ \bibinfo {pages} {8048} (\bibinfo
  {year} {2015})}\BibitemShut {NoStop}%
\bibitem [{\citenamefont {Verbanis}\ \emph {et~al.}(2016)\citenamefont
  {Verbanis}, \citenamefont {Martin}, \citenamefont {Rosset}, \citenamefont
  {Lim}, \citenamefont {Thew},\ and\ \citenamefont {Zbinden}}]{Verbanis2016}%
  \BibitemOpen
  \bibfield  {author} {\bibinfo {author} {\bibfnamefont {E.}~\bibnamefont
  {Verbanis}}, \bibinfo {author} {\bibfnamefont {A.}~\bibnamefont {Martin}},
  \bibinfo {author} {\bibfnamefont {D.}~\bibnamefont {Rosset}}, \bibinfo
  {author} {\bibfnamefont {C.~C.~W.}\ \bibnamefont {Lim}}, \bibinfo {author}
  {\bibfnamefont {R.~T.}\ \bibnamefont {Thew}}, \ and\ \bibinfo {author}
  {\bibfnamefont {H.}~\bibnamefont {Zbinden}},\ }\href {\doibase
  10.1103/PhysRevLett.116.190501} {\bibfield  {journal} {\bibinfo  {journal}
  {Phys. Rev. Lett.}\ }\textbf {\bibinfo {volume} {116}},\ \bibinfo {pages}
  {190501} (\bibinfo {year} {2016})}\BibitemShut {NoStop}%
\bibitem [{\citenamefont {Shahandeh}\ \emph {et~al.}(2017)\citenamefont
  {Shahandeh}, \citenamefont {Hall},\ and\ \citenamefont
  {Ralph}}]{Shahandeh2017}%
  \BibitemOpen
  \bibfield  {author} {\bibinfo {author} {\bibfnamefont {F.}~\bibnamefont
  {Shahandeh}}, \bibinfo {author} {\bibfnamefont {M.~J.~W.}\ \bibnamefont
  {Hall}}, \ and\ \bibinfo {author} {\bibfnamefont {T.~C.}\ \bibnamefont
  {Ralph}},\ }\href {\doibase 10.1103/PhysRevLett.118.150505} {\bibfield
  {journal} {\bibinfo  {journal} {Phys. Rev. Lett.}\ }\textbf {\bibinfo
  {volume} {118}},\ \bibinfo {pages} {150505} (\bibinfo {year}
  {2017})}\BibitemShut {NoStop}%
\bibitem [{\citenamefont {Chitambar}\ \emph {et~al.}(2014)\citenamefont
  {Chitambar}, \citenamefont {Leung}, \citenamefont {Man{\v{c}}inska},
  \citenamefont {Ozols},\ and\ \citenamefont {Winter}}]{Chitambar2014}%
  \BibitemOpen
  \bibfield  {author} {\bibinfo {author} {\bibfnamefont {E.}~\bibnamefont
  {Chitambar}}, \bibinfo {author} {\bibfnamefont {D.}~\bibnamefont {Leung}},
  \bibinfo {author} {\bibfnamefont {L.}~\bibnamefont {Man{\v{c}}inska}},
  \bibinfo {author} {\bibfnamefont {M.}~\bibnamefont {Ozols}}, \ and\ \bibinfo
  {author} {\bibfnamefont {A.}~\bibnamefont {Winter}},\ }\href {\doibase
  10.1007/s00220-014-1953-9} {\bibfield  {journal} {\bibinfo  {journal}
  {Commun. Math. Phys.}\ }\textbf {\bibinfo {volume} {328}},\ \bibinfo {pages}
  {303} (\bibinfo {year} {2014})}\BibitemShut {NoStop}%
\bibitem [{\citenamefont {Luenberger}\ and\ \citenamefont
  {Ye}(2016)}]{Luenberger2016}%
  \BibitemOpen
  \bibfield  {author} {\bibinfo {author} {\bibfnamefont {D.~G.}\ \bibnamefont
  {Luenberger}}\ and\ \bibinfo {author} {\bibfnamefont {Y.}~\bibnamefont
  {Ye}},\ }\enquote {\bibinfo {title} {Conic linear programming},}\ in\ \href
  {\doibase 10.1007/978-3-319-18842-3_6} {\emph {\bibinfo {booktitle} {Linear
  and Nonlinear Programming}}}\ (\bibinfo  {publisher} {Springer International
  Publishing},\ \bibinfo {address} {Cham},\ \bibinfo {year} {2016})\ pp.\
  \bibinfo {pages} {149--176}\BibitemShut {NoStop}%
\bibitem [{\citenamefont {Sturm}(2002)}]{Sturm2002}%
  \BibitemOpen
  \bibfield  {author} {\bibinfo {author} {\bibfnamefont {J.~F.}\ \bibnamefont
  {Sturm}},\ }\href {\doibase 10.1080/1055678021000045123} {\bibfield
  {journal} {\bibinfo  {journal} {Optim. Methods Softw.}\ }\textbf {\bibinfo
  {volume} {17}},\ \bibinfo {pages} {1105} (\bibinfo {year}
  {2002})}\BibitemShut {NoStop}%
\bibitem [{\citenamefont {T{\"{u}}t{\"{u}}nc{\"{u}}}\ \emph
  {et~al.}(2003)\citenamefont {T{\"{u}}t{\"{u}}nc{\"{u}}}, \citenamefont
  {Toh},\ and\ \citenamefont {Todd}}]{Tutuncu2003}%
  \BibitemOpen
  \bibfield  {author} {\bibinfo {author} {\bibfnamefont {R.~H.}\ \bibnamefont
  {T{\"{u}}t{\"{u}}nc{\"{u}}}}, \bibinfo {author} {\bibfnamefont {K.~C.}\
  \bibnamefont {Toh}}, \ and\ \bibinfo {author} {\bibfnamefont {M.~J.}\
  \bibnamefont {Todd}},\ }\href {\doibase 10.1007/s10107-002-0347-5} {\bibfield
   {journal} {\bibinfo  {journal} {Math. Program.}\ }\textbf {\bibinfo {volume}
  {95}},\ \bibinfo {pages} {189} (\bibinfo {year} {2003})}\BibitemShut
  {NoStop}%
\bibitem [{\citenamefont {Zukowski}\ \emph {et~al.}(1999)\citenamefont
  {Zukowski}, \citenamefont {Kaszlikowski}, \citenamefont {Baturo},\ and\
  \citenamefont {Larsson}}]{Zukowski1999a}%
  \BibitemOpen
  \bibfield  {author} {\bibinfo {author} {\bibfnamefont {M.}~\bibnamefont
  {Zukowski}}, \bibinfo {author} {\bibfnamefont {D.}~\bibnamefont
  {Kaszlikowski}}, \bibinfo {author} {\bibfnamefont {A.}~\bibnamefont
  {Baturo}}, \ and\ \bibinfo {author} {\bibfnamefont {J.-A.}\ \bibnamefont
  {Larsson}},\ }\href@noop {} {\  (\bibinfo {year} {1999})}\BibitemShut
  {NoStop}%
\bibitem [{\citenamefont {Kaszlikowski}\ \emph {et~al.}(2000)\citenamefont
  {Kaszlikowski}, \citenamefont {Gnaci{\'{n}}ski}, \citenamefont
  {{\.{Z}}ukowski}, \citenamefont {Miklaszewski},\ and\ \citenamefont
  {Zeilinger}}]{Kaszlikowski2000}%
  \BibitemOpen
  \bibfield  {author} {\bibinfo {author} {\bibfnamefont {D.}~\bibnamefont
  {Kaszlikowski}}, \bibinfo {author} {\bibfnamefont {P.}~\bibnamefont
  {Gnaci{\'{n}}ski}}, \bibinfo {author} {\bibfnamefont {M.}~\bibnamefont
  {{\.{Z}}ukowski}}, \bibinfo {author} {\bibfnamefont {W.}~\bibnamefont
  {Miklaszewski}}, \ and\ \bibinfo {author} {\bibfnamefont {A.}~\bibnamefont
  {Zeilinger}},\ }\href {\doibase 10.1103/PhysRevLett.85.4418} {\bibfield
  {journal} {\bibinfo  {journal} {Phys. Rev. Lett.}\ }\textbf {\bibinfo
  {volume} {85}},\ \bibinfo {pages} {4418} (\bibinfo {year}
  {2000})}\BibitemShut {NoStop}%
\bibitem [{\citenamefont {Lundeen}\ \emph {et~al.}(2009)\citenamefont
  {Lundeen}, \citenamefont {Feito}, \citenamefont {Coldenstrodt-Ronge},
  \citenamefont {Pregnell}, \citenamefont {Silberhorn}, \citenamefont {Ralph},
  \citenamefont {Eisert}, \citenamefont {Plenio},\ and\ \citenamefont
  {Walmsley}}]{Lundeen2009}%
  \BibitemOpen
  \bibfield  {author} {\bibinfo {author} {\bibfnamefont {J.~S.}\ \bibnamefont
  {Lundeen}}, \bibinfo {author} {\bibfnamefont {A.}~\bibnamefont {Feito}},
  \bibinfo {author} {\bibfnamefont {H.}~\bibnamefont {Coldenstrodt-Ronge}},
  \bibinfo {author} {\bibfnamefont {K.~L.}\ \bibnamefont {Pregnell}}, \bibinfo
  {author} {\bibfnamefont {C.}~\bibnamefont {Silberhorn}}, \bibinfo {author}
  {\bibfnamefont {T.~C.}\ \bibnamefont {Ralph}}, \bibinfo {author}
  {\bibfnamefont {J.}~\bibnamefont {Eisert}}, \bibinfo {author} {\bibfnamefont
  {M.~B.}\ \bibnamefont {Plenio}}, \ and\ \bibinfo {author} {\bibfnamefont
  {I.~a.}\ \bibnamefont {Walmsley}},\ }\href {\doibase 10.1038/nphys1133}
  {\bibfield  {journal} {\bibinfo  {journal} {Nat. Phys.}\ }\textbf {\bibinfo
  {volume} {5}},\ \bibinfo {pages} {27} (\bibinfo {year} {2009})}\BibitemShut
  {NoStop}%
\bibitem [{\citenamefont {Grant}\ and\ \citenamefont {Boyd}(2014)}]{Grant2014}%
  \BibitemOpen
  \bibfield  {author} {\bibinfo {author} {\bibfnamefont {M.}~\bibnamefont
  {Grant}}\ and\ \bibinfo {author} {\bibfnamefont {S.}~\bibnamefont {Boyd}},\
  }\href@noop {} {\enquote {\bibinfo {title} {{CVX: Matlab software for
  disciplined convex programming, version 2.1.}}}\ } (\bibinfo {year}
  {2014})\BibitemShut {NoStop}%
\bibitem [{\citenamefont {Lofberg}(2004)}]{Lofberg2004}%
  \BibitemOpen
  \bibfield  {author} {\bibinfo {author} {\bibfnamefont {J.}~\bibnamefont
  {Lofberg}},\ }in\ \href {\doibase 10.1109/CACSD.2004.1393890} {\emph
  {\bibinfo {booktitle} {2004 IEEE Int. Conf. Robot. Autom.}}}\ (\bibinfo
  {year} {2004})\ pp.\ \bibinfo {pages} {284--289}\BibitemShut {NoStop}%
\bibitem [{\citenamefont {Boyd}\ and\ \citenamefont
  {Vandenberghe}(2004)}]{Boyd2004}%
  \BibitemOpen
  \bibfield  {author} {\bibinfo {author} {\bibfnamefont {S.}~\bibnamefont
  {Boyd}}\ and\ \bibinfo {author} {\bibfnamefont {L.}~\bibnamefont
  {Vandenberghe}},\ }\href {\doibase 10.1017/CBO9780511804441} {\emph {\bibinfo
  {title} {Optim. Methods Softw.}}},\ Vol.~\bibinfo {volume} {25}\ (\bibinfo
  {year} {2004})\ pp.\ \bibinfo {pages} {487--487}\BibitemShut {NoStop}%
\bibitem [{\citenamefont {Lin}\ \emph {et~al.}()\citenamefont {Lin},
  \citenamefont {Rosset}, \citenamefont {Zhang}, \citenamefont {Bancal},\ and\
  \citenamefont {Liang}}]{Lin2017}%
  \BibitemOpen
  \bibfield  {author} {\bibinfo {author} {\bibfnamefont {P.-S.}\ \bibnamefont
  {Lin}}, \bibinfo {author} {\bibfnamefont {D.}~\bibnamefont {Rosset}},
  \bibinfo {author} {\bibfnamefont {Y.}~\bibnamefont {Zhang}}, \bibinfo
  {author} {\bibfnamefont {J.-D.}\ \bibnamefont {Bancal}}, \ and\ \bibinfo
  {author} {\bibfnamefont {Y.-C.}\ \bibnamefont {Liang}},\ }\href@noop {} {\
  }\Eprint {http://arxiv.org/abs/1705.09245} {arXiv:1705.09245} \BibitemShut
  {NoStop}%
\bibitem [{\citenamefont {Lim}(2016)}]{Lim2016}%
  \BibitemOpen
  \bibfield  {author} {\bibinfo {author} {\bibfnamefont {C.~C.~W.}\
  \bibnamefont {Lim}},\ }\href {\doibase 10.1103/PhysRevA.93.020101} {\bibfield
   {journal} {\bibinfo  {journal} {Phys. Rev. A}\ }\textbf {\bibinfo {volume}
  {93}},\ \bibinfo {pages} {020101} (\bibinfo {year} {2016})}\BibitemShut
  {NoStop}%
\bibitem [{\citenamefont {Gisin}(1996)}]{Gisin1996}%
  \BibitemOpen
  \bibfield  {author} {\bibinfo {author} {\bibfnamefont {N.}~\bibnamefont
  {Gisin}},\ }\href {\doibase 10.1016/S0375-9601(96)80001-6} {\bibfield
  {journal} {\bibinfo  {journal} {Phys. Lett. A}\ }\textbf {\bibinfo {volume}
  {210}},\ \bibinfo {pages} {151} (\bibinfo {year} {1996})}\BibitemShut
  {NoStop}%
\bibitem [{\citenamefont {Verstraete}\ \emph
  {et~al.}(2001{\natexlab{a}})\citenamefont {Verstraete}, \citenamefont
  {Dehaene},\ and\ \citenamefont {DeMoor}}]{Verstraete2001a}%
  \BibitemOpen
  \bibfield  {author} {\bibinfo {author} {\bibfnamefont {F.}~\bibnamefont
  {Verstraete}}, \bibinfo {author} {\bibfnamefont {J.}~\bibnamefont {Dehaene}},
  \ and\ \bibinfo {author} {\bibfnamefont {B.}~\bibnamefont {DeMoor}},\ }\href
  {\doibase 10.1103/PhysRevA.64.010101} {\bibfield  {journal} {\bibinfo
  {journal} {Phys. Rev. A}\ }\textbf {\bibinfo {volume} {64}},\ \bibinfo
  {pages} {010101} (\bibinfo {year} {2001}{\natexlab{a}})}\BibitemShut
  {NoStop}%
\bibitem [{\citenamefont {Verstraete}\ \emph
  {et~al.}(2001{\natexlab{b}})\citenamefont {Verstraete}, \citenamefont
  {Audenaert}, \citenamefont {Dehaene},\ and\ \citenamefont {{De
  Moor}}}]{Verstraete2001}%
  \BibitemOpen
  \bibfield  {author} {\bibinfo {author} {\bibfnamefont {F.}~\bibnamefont
  {Verstraete}}, \bibinfo {author} {\bibfnamefont {K.}~\bibnamefont
  {Audenaert}}, \bibinfo {author} {\bibfnamefont {J.}~\bibnamefont {Dehaene}},
  \ and\ \bibinfo {author} {\bibfnamefont {B.}~\bibnamefont {{De Moor}}},\
  }\href {\doibase 10.1088/0305-4470/34/47/329} {\bibfield  {journal} {\bibinfo
   {journal} {J. Phys. A. Math. Gen.}\ }\textbf {\bibinfo {volume} {34}},\
  \bibinfo {pages} {10327} (\bibinfo {year} {2001}{\natexlab{b}})}\BibitemShut
  {NoStop}%
\bibitem [{\citenamefont {Verstraete}\ and\ \citenamefont
  {Wolf}(2002)}]{Verstraete2002}%
  \BibitemOpen
  \bibfield  {author} {\bibinfo {author} {\bibfnamefont {F.}~\bibnamefont
  {Verstraete}}\ and\ \bibinfo {author} {\bibfnamefont {M.~M.}\ \bibnamefont
  {Wolf}},\ }\href {\doibase 10.1103/PhysRevLett.89.170401} {\bibfield
  {journal} {\bibinfo  {journal} {Phys. Rev. Lett.}\ }\textbf {\bibinfo
  {volume} {89}},\ \bibinfo {pages} {170401} (\bibinfo {year}
  {2002})}\BibitemShut {NoStop}%
\bibitem [{\citenamefont {Rosset}\ \emph {et~al.}(2017)\citenamefont {Rosset},
  \citenamefont {Buscemi},\ and\ \citenamefont {Liang}}]{Rosset2017}%
  \BibitemOpen
  \bibfield  {author} {\bibinfo {author} {\bibfnamefont {D.}~\bibnamefont
  {Rosset}}, \bibinfo {author} {\bibfnamefont {F.}~\bibnamefont {Buscemi}}, \
  and\ \bibinfo {author} {\bibfnamefont {Y.-C.}\ \bibnamefont {Liang}},\
  }\href@noop {} {\  (\bibinfo {year} {2017})},\ \Eprint
  {http://arxiv.org/abs/arXiv:1710.04710} {arXiv:1710.04710} \BibitemShut
  {NoStop}%
\bibitem [{\citenamefont {{MOSEK ApS.}}(2015)}]{ApS2015}%
  \BibitemOpen
  \bibfield  {author} {\bibinfo {author} {\bibnamefont {{MOSEK ApS.}}},\
  }\href@noop {} {\  (\bibinfo {year} {2015})}\BibitemShut {NoStop}%
\bibitem [{\citenamefont {Grant}\ \emph {et~al.}(2006)\citenamefont {Grant},
  \citenamefont {Boyd},\ and\ \citenamefont {Ye}}]{Grant2006}%
  \BibitemOpen
  \bibfield  {author} {\bibinfo {author} {\bibfnamefont {M.}~\bibnamefont
  {Grant}}, \bibinfo {author} {\bibfnamefont {S.}~\bibnamefont {Boyd}}, \ and\
  \bibinfo {author} {\bibfnamefont {Y.}~\bibnamefont {Ye}},\ }in\ \href
  {\doibase 10.1007/0-387-30528-9_7} {\emph {\bibinfo {booktitle} {Glob.
  Optim.}}},\ Vol.~\bibinfo {volume} {C}\ (\bibinfo {year} {2006})\ pp.\
  \bibinfo {pages} {155--210}\BibitemShut {NoStop}%
\bibitem [{\citenamefont {Chen}\ \emph {et~al.}(2016)\citenamefont {Chen},
  \citenamefont {Budroni}, \citenamefont {Liang},\ and\ \citenamefont
  {Chen}}]{Chen2016}%
  \BibitemOpen
  \bibfield  {author} {\bibinfo {author} {\bibfnamefont {S.-L.}\ \bibnamefont
  {Chen}}, \bibinfo {author} {\bibfnamefont {C.}~\bibnamefont {Budroni}},
  \bibinfo {author} {\bibfnamefont {Y.-C.}\ \bibnamefont {Liang}}, \ and\
  \bibinfo {author} {\bibfnamefont {Y.-N.}\ \bibnamefont {Chen}},\ }\href
  {\doibase 10.1103/PhysRevLett.116.240401} {\bibfield  {journal} {\bibinfo
  {journal} {Phys. Rev. Lett.}\ }\textbf {\bibinfo {volume} {116}},\ \bibinfo
  {pages} {240401} (\bibinfo {year} {2016})}\BibitemShut {NoStop}%
\bibitem [{\citenamefont {T{\'{o}}th}\ \emph {et~al.}(2015)\citenamefont
  {T{\'{o}}th}, \citenamefont {Moroder},\ and\ \citenamefont
  {G{\"{u}}hne}}]{Toth2015}%
  \BibitemOpen
  \bibfield  {author} {\bibinfo {author} {\bibfnamefont {G.}~\bibnamefont
  {T{\'{o}}th}}, \bibinfo {author} {\bibfnamefont {T.}~\bibnamefont {Moroder}},
  \ and\ \bibinfo {author} {\bibfnamefont {O.}~\bibnamefont {G{\"{u}}hne}},\
  }\href {\doibase 10.1103/PhysRevLett.114.160501} {\bibfield  {journal}
  {\bibinfo  {journal} {Phys. Rev. Lett.}\ }\textbf {\bibinfo {volume} {114}},\
  \bibinfo {pages} {160501} (\bibinfo {year} {2015})}\BibitemShut {NoStop}%
\bibitem [{\citenamefont {Bennett}\ \emph {et~al.}(1999)\citenamefont
  {Bennett}, \citenamefont {DiVincenzo}, \citenamefont {Fuchs}, \citenamefont
  {Mor}, \citenamefont {Rains}, \citenamefont {Shor}, \citenamefont {Smolin},\
  and\ \citenamefont {Wootters}}]{Bennett1999}%
  \BibitemOpen
  \bibfield  {author} {\bibinfo {author} {\bibfnamefont {C.~H.}\ \bibnamefont
  {Bennett}}, \bibinfo {author} {\bibfnamefont {D.~P.}\ \bibnamefont
  {DiVincenzo}}, \bibinfo {author} {\bibfnamefont {C.~A.}\ \bibnamefont
  {Fuchs}}, \bibinfo {author} {\bibfnamefont {T.}~\bibnamefont {Mor}}, \bibinfo
  {author} {\bibfnamefont {E.}~\bibnamefont {Rains}}, \bibinfo {author}
  {\bibfnamefont {P.~W.}\ \bibnamefont {Shor}}, \bibinfo {author}
  {\bibfnamefont {J.~A.}\ \bibnamefont {Smolin}}, \ and\ \bibinfo {author}
  {\bibfnamefont {W.~K.}\ \bibnamefont {Wootters}},\ }\href {\doibase
  10.1103/PhysRevA.59.1070} {\bibfield  {journal} {\bibinfo  {journal} {Phys.
  Rev. A}\ }\textbf {\bibinfo {volume} {59}},\ \bibinfo {pages} {1070}
  (\bibinfo {year} {1999})}\BibitemShut {NoStop}%
\bibitem [{\citenamefont {{\v{S}}upi{\'{c}}}\ \emph {et~al.}(2017)\citenamefont
  {{\v{S}}upi{\'{c}}}, \citenamefont {Skrzypczyk},\ and\ \citenamefont
  {Cavalcanti}}]{Supic2017a}%
  \BibitemOpen
  \bibfield  {author} {\bibinfo {author} {\bibfnamefont {I.}~\bibnamefont
  {{\v{S}}upi{\'{c}}}}, \bibinfo {author} {\bibfnamefont {P.}~\bibnamefont
  {Skrzypczyk}}, \ and\ \bibinfo {author} {\bibfnamefont {D.}~\bibnamefont
  {Cavalcanti}},\ }\href {\doibase 10.1103/PhysRevA.95.042340} {\bibfield
  {journal} {\bibinfo  {journal} {Phys. Rev. A}\ }\textbf {\bibinfo {volume}
  {95}},\ \bibinfo {pages} {042340} (\bibinfo {year} {2017})}\BibitemShut
  {NoStop}%
\bibitem [{\citenamefont {Albeverio}\ \emph {et~al.}(2002)\citenamefont
  {Albeverio}, \citenamefont {Fei},\ and\ \citenamefont
  {Yang}}]{Albeverio2002}%
  \BibitemOpen
  \bibfield  {author} {\bibinfo {author} {\bibfnamefont {S.}~\bibnamefont
  {Albeverio}}, \bibinfo {author} {\bibfnamefont {S.~M.}\ \bibnamefont {Fei}},
  \ and\ \bibinfo {author} {\bibfnamefont {W.~L.}\ \bibnamefont {Yang}},\
  }\href {\doibase DOI 10.1103/PhysRevA.66.012301} {\bibfield  {journal}
  {\bibinfo  {journal} {Phys. Rev. A}\ }\textbf {\bibinfo {volume} {66}},\
  \bibinfo {pages} {012301} (\bibinfo {year} {2002})}\BibitemShut {NoStop}%
\bibitem [{\citenamefont {Doherty}\ \emph {et~al.}(2004)\citenamefont
  {Doherty}, \citenamefont {Parrilo},\ and\ \citenamefont
  {Spedalieri}}]{Doherty2004}%
  \BibitemOpen
  \bibfield  {author} {\bibinfo {author} {\bibfnamefont {A.~C.}\ \bibnamefont
  {Doherty}}, \bibinfo {author} {\bibfnamefont {P.~A.}\ \bibnamefont
  {Parrilo}}, \ and\ \bibinfo {author} {\bibfnamefont {F.~M.}\ \bibnamefont
  {Spedalieri}},\ }\href {\doibase 10.1103/PhysRevA.69.022308} {\bibfield
  {journal} {\bibinfo  {journal} {Phys. Rev. A}\ }\textbf {\bibinfo {volume}
  {69}},\ \bibinfo {pages} {022308} (\bibinfo {year} {2004})}\BibitemShut
  {NoStop}%
\bibitem [{\citenamefont {Navascu{\'{e}}s}\ \emph {et~al.}(2008)\citenamefont
  {Navascu{\'{e}}s}, \citenamefont {Pironio},\ and\ \citenamefont
  {Ac{\'{i}}n}}]{Navascues2008a}%
  \BibitemOpen
  \bibfield  {author} {\bibinfo {author} {\bibfnamefont {M.}~\bibnamefont
  {Navascu{\'{e}}s}}, \bibinfo {author} {\bibfnamefont {S.}~\bibnamefont
  {Pironio}}, \ and\ \bibinfo {author} {\bibfnamefont {A.}~\bibnamefont
  {Ac{\'{i}}n}},\ }\href {\doibase 10.1088/1367-2630/10/7/073013} {\bibfield
  {journal} {\bibinfo  {journal} {New J. Phys.}\ }\textbf {\bibinfo {volume}
  {10}},\ \bibinfo {pages} {073013} (\bibinfo {year} {2008})}\BibitemShut
  {NoStop}%
\bibitem [{\citenamefont {Helton}\ and\ \citenamefont
  {Nie}(2010)}]{Helton2010}%
  \BibitemOpen
  \bibfield  {author} {\bibinfo {author} {\bibfnamefont {J.~W.}\ \bibnamefont
  {Helton}}\ and\ \bibinfo {author} {\bibfnamefont {J.}~\bibnamefont {Nie}},\
  }\href {\doibase 10.1007/s10107-008-0240-y} {\bibfield  {journal} {\bibinfo
  {journal} {Math. Program.}\ }\textbf {\bibinfo {volume} {122}},\ \bibinfo
  {pages} {21} (\bibinfo {year} {2010})}\BibitemShut {NoStop}%
\bibitem [{\citenamefont {O'Donoghue}\ \emph {et~al.}(2016)\citenamefont
  {O'Donoghue}, \citenamefont {Chu}, \citenamefont {Parikh},\ and\
  \citenamefont {Boyd}}]{ODonoghue2016}%
  \BibitemOpen
  \bibfield  {author} {\bibinfo {author} {\bibfnamefont {B.}~\bibnamefont
  {O'Donoghue}}, \bibinfo {author} {\bibfnamefont {E.}~\bibnamefont {Chu}},
  \bibinfo {author} {\bibfnamefont {N.}~\bibnamefont {Parikh}}, \ and\ \bibinfo
  {author} {\bibfnamefont {S.}~\bibnamefont {Boyd}},\ }\href@noop {} {\enquote
  {\bibinfo {title} {{SCS}: Splitting conic solver, version 1.2.6},}\ }\bibinfo
  {howpublished} {\url{https://github.com/cvxgrp/scs}} (\bibinfo {year}
  {2016})\BibitemShut {NoStop}%
\bibitem [{\citenamefont {Brand{\~{a}}o}\ and\ \citenamefont
  {Vianna}(2004)}]{Brandao2004}%
  \BibitemOpen
  \bibfield  {author} {\bibinfo {author} {\bibfnamefont {F.~G. S.~L.}\
  \bibnamefont {Brand{\~{a}}o}}\ and\ \bibinfo {author} {\bibfnamefont {R.~O.}\
  \bibnamefont {Vianna}},\ }\href {\doibase 10.1103/PhysRevA.70.062309}
  {\bibfield  {journal} {\bibinfo  {journal} {Phys. Rev. A}\ }\textbf {\bibinfo
  {volume} {70}},\ \bibinfo {pages} {062309} (\bibinfo {year}
  {2004})}\BibitemShut {NoStop}%
\bibitem [{\citenamefont {Jafarizadeh}\ \emph {et~al.}(2005)\citenamefont
  {Jafarizadeh}, \citenamefont {Mirzaee},\ and\ \citenamefont
  {Rezaee}}]{Jafarizadeh2005}%
  \BibitemOpen
  \bibfield  {author} {\bibinfo {author} {\bibfnamefont {M.~A.}\ \bibnamefont
  {Jafarizadeh}}, \bibinfo {author} {\bibfnamefont {M.}~\bibnamefont
  {Mirzaee}}, \ and\ \bibinfo {author} {\bibfnamefont {M.}~\bibnamefont
  {Rezaee}},\ }\href {\doibase 10.1142/S0219749905001043} {\bibfield  {journal}
  {\bibinfo  {journal} {Int. J. Quantum Inf.}\ }\textbf {\bibinfo {volume}
  {03}},\ \bibinfo {pages} {511} (\bibinfo {year} {2005})}\BibitemShut
  {NoStop}%
\bibitem [{\citenamefont {Grant}\ and\ \citenamefont {Boyd}(2008)}]{Grant2008}%
  \BibitemOpen
  \bibfield  {author} {\bibinfo {author} {\bibfnamefont {M.~C.}\ \bibnamefont
  {Grant}}\ and\ \bibinfo {author} {\bibfnamefont {S.~P.}\ \bibnamefont
  {Boyd}},\ }in\ \href {\doibase 10.1007/978-1-84800-155-8_7} {\emph {\bibinfo
  {booktitle} {Recent Adv. Learn. Control}}},\ Vol.\ \bibinfo {volume} {371}\
  (\bibinfo {year} {2008})\ pp.\ \bibinfo {pages} {95--110}\BibitemShut
  {NoStop}%
\bibitem [{\citenamefont {Plenio}(2005)}]{Plenio2005}%
  \BibitemOpen
  \bibfield  {author} {\bibinfo {author} {\bibfnamefont {M.~B.}\ \bibnamefont
  {Plenio}},\ }\href {\doibase 10.1103/PhysRevLett.95.090503} {\bibfield
  {journal} {\bibinfo  {journal} {Phys. Rev. Lett.}\ }\textbf {\bibinfo
  {volume} {95}},\ \bibinfo {pages} {090503} (\bibinfo {year}
  {2005})}\BibitemShut {NoStop}%
\bibitem [{\citenamefont {Horodecki}\ \emph {et~al.}(1996)\citenamefont
  {Horodecki}, \citenamefont {Horodecki},\ and\ \citenamefont
  {Horodecki}}]{Horodecki1996}%
  \BibitemOpen
  \bibfield  {author} {\bibinfo {author} {\bibfnamefont {M.}~\bibnamefont
  {Horodecki}}, \bibinfo {author} {\bibfnamefont {P.}~\bibnamefont
  {Horodecki}}, \ and\ \bibinfo {author} {\bibfnamefont {R.}~\bibnamefont
  {Horodecki}},\ }\href {\doibase 10.1016/S0375-9601(96)00706-2} {\bibfield
  {journal} {\bibinfo  {journal} {Phys. Lett. A}\ }\textbf {\bibinfo {volume}
  {223}},\ \bibinfo {pages} {1} (\bibinfo {year} {1996})}\BibitemShut {NoStop}%
\bibitem [{\citenamefont {Gurvits}(2004)}]{Gurvits2004}%
  \BibitemOpen
  \bibfield  {author} {\bibinfo {author} {\bibfnamefont {L.}~\bibnamefont
  {Gurvits}},\ }\href {\doibase 10.1016/j.jcss.2004.06.003} {\bibfield
  {journal} {\bibinfo  {journal} {J. Comput. Syst. Sci.}\ }\textbf {\bibinfo
  {volume} {69}},\ \bibinfo {pages} {448} (\bibinfo {year} {2004})}\BibitemShut
  {NoStop}%
\bibitem [{\citenamefont {Hildebrand}(2008)}]{Hildebrand2008}%
  \BibitemOpen
  \bibfield  {author} {\bibinfo {author} {\bibfnamefont {R.}~\bibnamefont
  {Hildebrand}},\ }\href {\doibase http://dx.doi.org/10.1016/j.laa.2008.04.018}
  {\bibfield  {journal} {\bibinfo  {journal} {Lin. Alg. and its Applications}\
  }\textbf {\bibinfo {volume} {429}},\ \bibinfo {pages} {901} (\bibinfo {year}
  {2008})}\BibitemShut {NoStop}%
\bibitem [{\citenamefont {Bandyopadhyay}\ \emph {et~al.}(2009)\citenamefont
  {Bandyopadhyay}, \citenamefont {Brassard}, \citenamefont {Kimmel},\ and\
  \citenamefont {Wootters}}]{Bandyopadhyay2009}%
  \BibitemOpen
  \bibfield  {author} {\bibinfo {author} {\bibfnamefont {S.}~\bibnamefont
  {Bandyopadhyay}}, \bibinfo {author} {\bibfnamefont {G.}~\bibnamefont
  {Brassard}}, \bibinfo {author} {\bibfnamefont {S.}~\bibnamefont {Kimmel}}, \
  and\ \bibinfo {author} {\bibfnamefont {W.~K.}\ \bibnamefont {Wootters}},\
  }\href {\doibase 10.1103/PhysRevA.80.012313} {\bibfield  {journal} {\bibinfo
  {journal} {Phys. Rev. A}\ }\textbf {\bibinfo {volume} {80}},\ \bibinfo
  {pages} {012313} (\bibinfo {year} {2009})}\BibitemShut {NoStop}%
\end{thebibliography}%

\appendix

\section{Proof of Proposition~\ref{Prop:OptimalMeasurements}}
\label{App:OptimalMeasurements}
The proof rests on quantum teleportation: the generalized Bell measurements allow perfect teleportation of $\rhoAB$ into the recovered states $\mu_{a b}$, up to local unitaries~\footnote{For qubits, we recover the usual Bell measurement in the definition below by taking the unitaries $U_a$, $V_b$ from the set $\{ \mathbbm{1}, \sigma_x, \sigma_y, \sigma_z \}$ containing the identity and the three Pauli matrices.}.
The elements $\{ A_a \}$ and $\{B_b \}$ of this measurement can be written~{\cite{Albeverio2002}} $A_a \equiv \ketbra{\mathcal{A}_a}$, $B_b \equiv \ketbra{\mathcal{B}_b}$ with 
\begin{equation} \label{Eq:GeneralizedBellMeasurement} 
 \ket{\mathcal{A}_a} \equiv (U_a \otimes \mathbbm{1}) \ket{\varphi_{\dX}}, \qquad \ket{\mathcal{B}_b} \equiv (\mathbbm{1} \otimes V_b) \ket{\varphi_{\dY}},
\end{equation}
where
\begin{equation}
 \ket{\varphi_d} \equiv \frac{1}{\sqrt{d}} \sum_{i = 1}^d \ket{i i}
\end{equation}
and $U_a$, $V_b$ are unitary operators for $a = 1 \ldots \nA$ and $b = 1 \ldots \nB$.
As the sets of inputs are tomographically complete, the correlations $P (a b | x y)$ are in one-to-one correspondence to this effective POVM $\{ \Pi_{a b} \}$.
The effective POVM elements~(\ref{Eq:EffectivePOVMQuantum}) are: 
\begin{equation}
 \Pi_{a b} = \frac{1}{\dX \dY} (U_a \otimes V_b) \rhoAB^{\top} (U_a^{\dag} \otimes V_b^{\dag}),
\end{equation}
and the recovered states are~(\ref{Eq:RecoveredStates}):
\begin{multline} \label{Eq:RecoveredStatesBell}
 \mu_{a b} = \frac{\Pi_{a b}^{\top}}{\tmop{tr} [\Pi_{a b}]} = \left( \overline{U}_a \otimes \overline{V}_b \right) \rhoAB \left( \overline{U}_a \otimes \overline{V}_b \right)^{\dag}, \\
  p_{a b} = \frac{1}{\dX^2 \dY^2} .
\end{multline}
In the above, we wrote $U^{\dag}$ for the conjugate transpose and $\overline{U}$ for the complex conjugate of $U$.
Then, as any entanglement measure is invariant under local unitaries, $\mathcal{E} (\mu_{a b}) =\mathcal{E} \left( \rhoAB \right)$; when the devices implement full Bell measurements on $\nA = \dX^2$ and $\nB = \dY^2$ outcomes, the bound $\nu^{\ast}$ in~(\ref{Eq:LowerBound}) is tight: $\nu^{\ast} =\mathcal{E} \left( \rhoAB \right)$.

\section{Conic linear programs}\label{App:CLP}

Linear~{\cite{Brunner2014}} and semidefinite~{\cite{Doherty2004,Navascues2008a}} programs are widely used in quantum information, and are part of the larger family of {\textit{conic linear programs}}~{\cite{Luenberger2016}}.
In their standard form, we write:
\begin{equation}  \label{Eq:CLP} 
 \arraycolsep=2pt\def\arraystretch{1.8}
 \begin{array}{ll || ll}
\multicolumn{4}{c}{ \text{{\tmstrong{Conic linear programming problem}}} } \\
 & \text{Primal} & &  \text{Dual}\\
 \text{minimize } & \vec{c}^{\top} \vec{x} & \,& \vec{b}^{\top} \vec{y}\\
 \text{over } & \vec{x} \in \mathcal{K} \subseteq \mathbbm{R}^n & \qquad & \vec{y} \in \mathbbm{R}^m\\
 \text{subject to } & A \vec{x} = \vec{b} & & \vec{c} - A^{\top} \vec{y} \in \mathcal{K}^{\ast}
 \end{array}
\end{equation}
where $\mathcal{K}$ is a convex cone with dual cone
\begin{equation}
  \mathcal{K}^{\ast} = \{ \vec{z} \in \mathbbm{R}^n \text{ s.t. } \vec{x}^\top \vec{z} \ge 0 \text{ for all } \vec{x} \in \mathcal{K} \}
\end{equation}
with $A \in \mathbbm{R}^{m \times n}$, $\vec{b} \in \mathbbm{R}^m$ and $\vec{c} \in \mathbbm{R}^n$.
In this formulation, $\vec{x}$ is a real vector; the cone of semidefinite positive complex or real matrices is handled by a representing these matrices over a real basis.
The problem is fully specified by the data $A$, $\vec{b}$, $\vec{c}$ and the structure of the cone $\mathcal{K}$.
The optimal solution is given by the pair $(\vec{x}^{\ast}, \vec{y}^{\ast})$, and $p^{\ast} = \vec{c}^{\top} \vec{x}^{\ast}$ and $d^{\ast} = \vec{b}^{\top} \vec{y}^{\ast}$ are respectively the primal and dual objective.
Weak duality states that
\begin{equation}
 \label{Eq:WeakDuality} p^{\ast} \geqslant d^{\ast}
\end{equation}
always holds.
However, when the problem is strictly feasible (Slater condition), a stronger statement holds (strong duality):
\begin{equation}
 \label{Eq:StrongDuality} p^{\ast} = d^{\ast} .
\end{equation}
The formulations presented in the Appendices~\ref{App:RecipesEntanglement} and~\ref{App:RecipesSep} maintain strict feasibility, and we always have $p^{\ast} = d^{\ast}$ in practice.

The conic linear form allows us to hide eventual complexities in the formulation.
The definition of $\mathcal{K}$ can include additional equality constraints $E \vec{x} = 0$, the declaration of additional variables (for example, when projecting semidefinite cones~{\cite{Helton2010}}); as long as the set $\mathcal{K}$ thus represented is a cone, the interpretation of~(\ref{Eq:CLP}) as a primal-dual pair holds.
In particular, the dual objective only depends on the dual variables associated with the primal constraint $A \vec{x} = \vec{b}$; and in the dual problem, $\vec{b}$ is only present in the objective.

Numerical solvers handle generally a cone $\mathcal{K}=\mathcal{K}_1 \times \mathcal{K}_2 \ldots \times \mathcal{K}_p$ that is the Cartesian product of basic supported cones.
For example, SeDuMi~{\cite{Sturm2002}}, SDPT3~{\cite{Tutuncu2003}} and Mosek~{\cite{ApS2015}} handle nonnegative, second order and semidefinite cones, while SCS~{\cite{ODonoghue2016}} also supports exponential and power cones.
However, the finer details of the reformulation can be delegated to a toolbox such as YALMIP~{\cite{Lofberg2004}} or CVX~{\cite{Grant2014}}.
These libraries allow the construction of a conic linear program using standard mathematical notation, and handle the transformation to the solver canonical input form automatically.
After handing a reformulated problem to the solver, these libraries transform back the solver output into the solution of the original problem.

Our entanglement measure cones are convenient mathematical objects which are not supported directly by YALMIP and CVX.
We thus present in Appendix~\ref{App:RecipesEntanglement} and~\ref{App:RecipesSep} formulations that are directly compatible with these libraries.
All the cones and conic linear programs presented in the Appendices map directly to the primal form of~(\ref{Eq:CLP}).

\section{Recipes: cones for entanglement measures}\label{App:RecipesEntanglement}

Let $\mathcal{E}$ be an entanglement measure satisfying axioms i and ii of Section~\ref{Sec:EntanglementMeasures}, whose the domain of validity has been extended to unnormalized states~(\ref{Eq:ExtendValidity}).
We verify that $\hat{\mathcal{E}}$ defined in~(\ref{Eq:EntanglementMeasureCone}) is a convex cone~{\cite{Boyd2004}}.
We write:
\begin{equation}
 \hat{\mathcal{E}} = \left\{ (\omega, \rho) \text{ such that } \tmop{tr} [\rho] \mathcal{E} \left( \frac{\rho}{\tmop{tr} [\rho]} \right) \leqslant \omega \right\} .
\end{equation}
Let us consider $ (\omega_1, \rho_1), (\omega_2, \rho_2) \in \hat{\mathcal{E}}$ and $\theta_1, \theta_2 \geqslant 0$; we need to show that $(\theta_1 \omega_1 + \theta_2 \omega_2, \theta_1 \rho_1 + \theta_2 \rho_2) \in \hat{\mathcal{E}}$.
With $t \equiv \tmop{tr} [\theta_1 \rho_1 + \theta_2 \rho_2]$, and using convexity of $\mathcal{E}$:
\begin{multline}
 t\mathcal{E} \left( \frac{\theta_1 \rho_1 + \theta_2 \rho_2}{t} \right) \\
 =  t\mathcal{E} \left( \frac{\tmop{tr} [\theta_1 \rho_1]}{t} \frac{\theta_1 \rho_1}{\tmop{tr} [\theta_1 \rho_1]} + \frac{\tmop{tr} [\theta_2 \rho_2]}{t} \frac{\theta_2 \rho_2}{\tmop{tr} [\theta_2 \rho_2]} \right)\\
 \leqslant \theta_1 \tmop{tr} [\rho_1] \mathcal{E} \left( \frac{\rho_1}{\tmop{tr} [\rho_1]} \right) + \theta_2 \tmop{tr} [\rho_2] \mathcal{E} \left( \frac{\rho_2}{\tmop{tr} [\rho_2]} \right)\\
 \leqslant \theta_1 \omega_1 + \theta_2 \omega_2 .
\end{multline}
All the measures mentioned in the main text use semidefinite cones, as already noted by various authors~{\cite{Brandao2004,Jafarizadeh2005,Moroder2013,Wang2016}}.
Building upon the conic linear form detailed in Appendix~\ref{App:CLP}, we provide below practical recipes to construct entanglement measure cones.
These formulations are inspired by the epigraph representations~{\cite{Grant2008}} used in disciplined convex programming, and can be entered using the standard syntax of YALMIP~{\cite{Lofberg2004}} and CVX\footnote{CVX has a particularly convenient syntax for the formulation of convex sets through \tmverbatim{cvx\_begin set sdp}.}~{\cite{Grant2014}}.
Our formulations always extend the domain of validity of the measures to unnormalized states.

Below, we have $\rhoAB \in \mathsf{Herm}_+ (\mathcal{A} \otimes \mathcal{B})$ for Hilbert spaces $\mathcal{A}$, $\mathcal{B}$ of dimension $d_{\mathcal{A}}$, $d_{\mathcal{B}}$ with $d \equiv d_{\mathcal{A}} d_{\mathcal{B}}$.
Note that the element $(\omega, \mathbbm{1}_{\mathcal{A}} \otimes \mathbbm{1}_{\mathcal{B}})$ is always in the interior of $\hat{\mathcal{E}}$ for $\omega > 0$, and thus our formulations preserve strict feasibility.

\subsection{Negativity}
\label{App:Recipe:Negativity}

The negativity $\mathcal{E}_{\text{NEG}}$ was proposed as a computable entanglement measure by Vidal and Werner~{\cite{Vidal2002}}.
The negativity is convex (but not the {\textit{logarithmic}} negativity~{\cite{Plenio2005}}) and dimension independent~{\cite{Eltschka2014}}.
For normalized states $\tmop{tr} [\rho] = 1$:
\begin{equation}
 \mathcal{E}_{\text{NEG}} (\rho) \equiv \frac{\| \rho^{\top_{\mathcal{A}}} \|_1 - 1}{2} \;,
\end{equation}
which corresponds to the absolute value of the sum of the negative eigenvalues of $\rho^{\top_\mathcal{A}}$.

Inspired by~{\cite{Moroder2013}}, we obtain a formulation using semidefinite cones by splitting $\rho^{\top_{\text{A}}} = \sigma^+ + (- \sigma^-)$, such that $\sigma^\pm$ are both semidefinite positive.
Then $\mathcal{E}_{\text{NEG}}(\rho) \le \tmop{tr}[\sigma^-]$, and the bound is tight when $\sigma^\pm$ have together the same eigenvalues as $\rho^{\top_{\text{A}}}$.
We obtain:

\begin{equation}
  \label{Eq:NegativityCone}
\arraycolsep=2pt\def\arraystretch{1.8}
 \begin{array}{ll}
\multicolumn{2}{c}{\text{{\tmstrong{Negativity cone}} } \hat{\mathcal{E}}_{\text{NEG}} } \\
 \text{We satisfy } & (\omega, \rho) \in \hat{\mathcal{E}}_{\text{NEG}}\\
 \text{iff there exists } & \omega \in \mathbbm{R}_+,\\
 & \rho, \sigma^+, \sigma^- \in \mathsf{Herm}_+ \left( \HA \otimes \HB
 \right)\\
 \text{such that } & \omega = \tmop{tr} [\sigma^-]\\
 & \rho^{\top_{\text{A}}} = \sigma^+ - \sigma^-\;.
 \end{array}
\end{equation}

\subsection{Robustness measures}

Vidal and Tarrach~{\cite{Vidal1999}} considered the robustness of $\rho$ relative to $\tau$ defined as:
\begin{equation}
 \label{Eq:Robustness} R (\rho \| \tau) = \min s \text{\quad such that \ } \rho + s \tau \in \mathsf{Sep} (\mathcal{A}: \mathcal{B}),
\end{equation}
where $\Sep{\mathcal{A}}{\mathcal{B}}$ is the cone of separable operators described in Appendix~\ref{App:RecipesSep}.
When $d \leqslant 6$, this cone is exactly described by the PPT criterion~{\cite{Horodecki1996}} (see~(\ref{Eq:SeparableConeExact}).
When $d > 6$, we replace $\mathsf{Sep} (\mathcal{A}: \mathcal{B})$ by one of the outer approximations~(\ref{Eq:SeparableConeOuter}), which provides an outer approximation of the entanglement measure cone.
When used in a conic linear program to quantify entanglement, we get a lower bound on the exact value.
From numerical evidence, we conjecture that the first level of approximation is always sufficient for the evaluation of these measures on pure states.
In Eq.~\eqref{Eq:Robustness}, we dropped the normalization factor $1 / (1 + s)$ as the separable cone is scale invariant.

Several entanglement measures can be constructed from the robustness, depending on the constraints on $\tau$.
By using a conic form and matching the original definition for $\tmop{tr} [\rho] = 1$, we extend automatically the domain of validity to $\tmop{tr} [\rho] \neq 1$.

\subsubsection{Random robustness}
\label{App:Recipe:RandomRobustness}

The random robustness~{\cite{Vidal1999}} $\mathcal{E}_{\text{RR}}$ minimizes $R (\rho \| \tau)$ with respect to $\tau = \mathbbm{1} / d$ the maximally random state.
It is not dimension independent.
\begin{equation}
 \arraycolsep=2pt\def\arraystretch{1.8}
 \begin{array}{ll}
 \multicolumn{2}{c}{\text{{\tmstrong{Random robustness cone}} } \hat{\mathcal{E}}_{\text{RR}}}\\
 \text{We satisfy } & (\omega, \rho) \in \hat{\mathcal{E}}_{\text{RR}}\\
 \text{iff there exists } & \omega \geqslant 0\\
 & \rho \in \mathsf{Herm}_+ (\mathcal{A} \otimes \mathcal{B})\\
 & \sigma \in \mathsf{Sep} (\mathcal{A}: \mathcal{B})\\
 \text{such that } & d \rho + \omega \mathbbm{1} = \sigma \; .
 \end{array}
\end{equation}

\subsubsection{Absolute robustness}
\label{App:Recipe:AbsoluteRobustness}

The absolute robustness~{\cite{Vidal1999}} $\mathcal{E}_{\text{AR}}$ is defined in a similar way; however, the minimization of the robustness is done over all separable states $\mathsf{Sep} (\mathcal{A}: \mathcal{B})$.
It is dimension independent.
In the formulation below, we substituted $\\tau = s \tau$ in the definition~(\ref{Eq:Robustness}).
\begin{equation}
 \arraycolsep=2pt\def\arraystretch{1.8}
 \begin{array}{ll}
   \multicolumn{2}{c}{\text{{\tmstrong{Absolute robustness cone}} }
   \hat{\mathcal{E}}_{\text{AR}}} \\
 \text{We satisfy } & (\omega, \rho) \in \hat{\mathcal{E}}_{\text{AR}}\\
 \text{iff there exists } & \omega \in \mathbbm{R}_+\\
 & \rho \in \mathsf{Herm}_+ (\mathcal{A} \otimes \mathcal{B})\\
 & \sigma, \tau \in \mathsf{Sep} (\mathcal{A}: \mathcal{B})\\
 \text{such that } & \omega = \tmop{tr} [\tau]\\
 & \rho + \tau = \sigma \; .
 \end{array}
\end{equation}

\subsubsection{Generalized robustness}
\label{App:Recipe:GeneralizedRobustness}

The generalized robustness is a variant introduced by Steiner~{\cite{Steiner2003}}, and minimizes the robustness over {\textit{all}} states.
It is dimension independent.
\begin{equation}
 \arraycolsep=2pt\def\arraystretch{1.8}
 \begin{array}{ll}
 \multicolumn{2}{c}{\text{{\tmstrong{Generalized robustness cone}} }\hat{\mathcal{E}}_{\text{GR}} } \\
 \text{We satisfy } & (\omega, \rho) \in \hat{\mathcal{E}}_{\text{GR}}\\
 \text{iff there exists } & \omega \in \mathbbm{R}_+\\
 & \rho, \tau \in \mathsf{Herm}_+ (\mathcal{A} \otimes \mathcal{B})\\
 & \sigma \in \mathsf{Sep} (\mathcal{A}: \mathcal{B})\\
 \text{such that } & \omega = \tmop{tr} [\tau]\\
 & \rho + \tau = \sigma \;.
 \end{array}
\end{equation}

\subsection{Semidefinite upper bound on distillable entanglement (DUB)}\label{App:Recipe:DUB}

Another entanglement measure computable using semidefinite cones was recently introduced by Wang and Duan~{\cite{Wang2016}}.
It provides an upper bound on the entanglement of distillation by LOCC~\cite{Plenio2007}.
\begin{equation}
 \arraycolsep=2pt\def\arraystretch{1.8}
 \begin{array}{ll}
 \multicolumn{2}{c}{\text{{\tmstrong{DUB cone}} } \hat{\mathcal{E}}_{\tmop{DUB}}} \\
 \text{We satisfy } & (\omega, \rho) \in \hat{\mathcal{E}}_{\text{DUB}}\\
 \text{iff there exists } & \omega \in \mathbbm{R}_+\\
 & \rho, U, V, \sigma \in \mathsf{Herm}_+ (\mathcal{A} \otimes
 \mathcal{B})\\
 \text{such that } & \omega = \tmop{tr} [U + V]\\
 & (U - V)^{\top_{\mathcal{B}}} - \rho = \sigma \;.
 \end{array}
\end{equation}
Note that the entanglement measure is not given directly by $\omega$, rather by $\log_2 \omega$.
From numerical observations, we conjecture that it is dimension independent.

\section{Recipes: separable cones}\label{App:RecipesSep}

For a bipartite Hilbert space $\tmH =\mathcal{A} \otimes \mathcal{B}$, we write $\Sep{\mathcal{A}}{\mathcal{B}}$ the set of separable operators:
\begin{equation}
  \sigma \in \Sep{\mathcal{A}}{\mathcal{B}} \qquad \Leftrightarrow \qquad \sigma = \sum_{\lambda} A_{\lambda} \otimes B_{\lambda},
\end{equation}
for some $A_{\lambda} \in \mathsf{Herm}_+ (\mathcal{A}), B_{\lambda} \in \mathsf{Herm}_+ (\mathcal{B})$.
Separable states are then represented by $\sigma \in \Sep{\mathcal{A}}{\mathcal{B}}$ with $\tmop{tr} [\sigma] = 1$ in addition.
For $\sigma \in \mathsf{Herm}_+ (\mathcal{A} \otimes \mathcal{B})$, we write $\sigma^{\top_{\mathcal{A}}}$ the partial transpose on $\mathcal{A}$.

We present now a convenient formulation of the cone $\mathsf{Sep} (\mathcal{A}: \mathcal{B})$ of separable operators on Hilbert spaces $\mathcal{A}$, $\mathcal{B}$ of finite dimension $\dA$, $\dB$, compatible with the framework of Appendix~\ref{App:CLP}.
The following semidefinite representation is exact~{\cite{Horodecki1996}} for $\dA \dB \leqslant 6$.
\begin{equation}\label{Eq:SeparableConeExact} 
 \arraycolsep=2pt\def\arraystretch{1.8}
 \begin{array}{ll}
 \multicolumn{2}{c}{\text{{\tmstrong{Separable cone}} } \mathsf{Sep} (\mathcal{A}: \mathcal{B}), \quad \dA \dB \leqslant 6 } \\
 \text{We satisfy } & \rho \in \Sep{\mathcal{A}}{\mathcal{B}}\\
 \text{iff there exists } & \rho \in \mathsf{Herm}_+ (\mathcal{A} \otimes
 \mathcal{B})\\
 & \sigma \in \mathsf{Herm}_+ (\mathcal{A} \otimes \mathcal{B})\\
 \text{such that } & \sigma^{\top_{\mathcal{A}}} = \rho \;.
 \end{array}
\end{equation}
The set of separable states has nonempty interior~{\cite{Vidal1999}}, and the same holds for $\mathsf{Sep} (\mathcal{A}: \mathcal{B})$.
Thus our formulations maintain strict feasibility.
For higher dimensions, the characterization of $\mathsf{Sep} (\mathcal{A}: \mathcal{B})$ is hard~{\cite{Gurvits2004,Hildebrand2008,Huang2014}}.
We use the semidefinite hierarchy proposed by Doherty et al.~{\cite{Doherty2004}} which represents a family of larger cones $\widetilde{\mathsf{S}}_k (\mathcal{A}: \mathcal{B})$ converging to $\mathsf{Sep} (\mathcal{A}: \mathcal{B})$ for $k \rightarrow \infty$.
These formulations still satisfy the framework of Appendix~\ref{App:CLP}.
For a symmetric extension involving $k = 2$ copies, with PPT constraints:
\begin{equation}  \label{Eq:SeparableConeOuter} 
 \arraycolsep=2pt\def\arraystretch{1.8}
 \begin{array}{ll}
 \multicolumn{2}{c}{\text{{\tmstrong{Separable cone (outer approximation)}} } \widetilde{\mathsf{S}}_2 (\mathcal{A}: \mathcal{B}) } \\
 \text{We satisfy } & \rho \in \widetilde{\mathsf{S}}_3 (\mathcal{A}: \mathcal{B})\\
 \text{iff there exists } & \tau \in \mathsf{Herm}_+ (\mathcal{A} \otimes \mathcal{B} \otimes \mathcal{B}')\\
 & \sigma_1, \sigma_2 \in \mathsf{Herm}_+ (\mathcal{A} \otimes \mathcal{B} \otimes \mathcal{B}')\\
 \text{such that } & \tmop{tr}_{\mathcal{B}'} [\tau] = \rho\\
 & \Pi_{\mathcal{B}\mathcal{B}'} \tau = \tau \Pi_{\mathcal{B}\mathcal{B}'}\\
 & \sigma_1 = \tau^{\top_{\mathcal{B}}}\\
 & \sigma_2 = \tau^{\top_{\mathcal{B}} \top_{\mathcal{B}'}}
 \end{array}
\end{equation}
where $\mathcal{B}'$ is isomorphic to $\mathcal{B}$ and $\Pi_{\mathcal{B}\mathcal{B}'}$ is the swap operator between $\mathcal{B}$ and $\mathcal{B}'$.

\section{Relation to our previous experimental work}
\label{App:RelationPreviousWork}
When using our construction with toolboxes such as CVX~{\cite{Grant2014}} and YALMIP~{\cite{Lofberg2004}}, the reformulation of the problem as a semidefinite program is done automatically.
Here, we perform this step manually and prove the validity of the program reported in the Supplemental Material of our Letter~{\cite{Verbanis2016}}, where we computed a quantitative MDIEW to lower bound the negativity present in a semiquantum scenario.
We use the recipe for the negativity entanglement cone $\hat{\mathcal{E}}_{\text{NEG}}$ given in Appendix~\ref{App:Recipe:Negativity} to express the constraint $(\omega, \rho) \in \hat{\mathcal{E}}_{\text{NEG}}$ in the entanglement quantification program~(\ref{Eq:ClaimCLPDuplicate}), introducing extra variables and constraints in the original program:
\begin{equation}
\arraycolsep=2pt\def\arraystretch{1.8}
 \label{Eq:NegativitySDP} \begin{array}{llll}
\multicolumn{2}{c}{ \text{{\tmstrong{Negativity quantification semidefinite }}}} \\[-0.8em]
\multicolumn{2}{c}{ \text{{\tmstrong{program}}(primal)}} \\
 \text{minimize } & \nu  = \frac{1}{\dX \dY} \sum_{a b} \nu_{a b}  \\
 \text{over } & \nu_{a b} \in  \mathbbm{R}_+, \forall a b\\
 & \Pi_{a b} \in  \mathsf{Herm}_+ \left( \HX \otimes \HY \right),
 \forall a b\\
 \text{{\textit{(extra)}}} & \sigma_{a b}^+, \sigma^-_{a b} \in \mathsf{Herm}_+ \left( \HX \otimes \HY \right), \forall a b\\
 \text{subject to } & \tmop{tr} [\Pi_{a b} \cdot (\xi_x \otimes \psi_y)] = P (a b | x y),\\
 \multicolumn{2}{r}{ \forall a b, (x, y) \in \mathcal{I}} \\
 \text{{\textit{(extra)}}} & \nu_{a b} = \tmop{tr} [\sigma^-_{a b}] , \forall a b\\
 & (\Pi_{a b})^{\top_{\text{X}}} = \sigma^+_{a b} - \sigma^-_{a b},  \forall a b,
 \end{array}
\end{equation}
which is equivalent (up to a factor $\dX \dY$) to the program presented in the Supplemental Material of~{\cite{Verbanis2016}}.
It is a semidefinite program in the (primal\footnote{According to the convention used in Appendix~\ref{App:CLP}.})canonical form.
By definition~(\ref{Eq:CLP}), the dual objective includes only the dual variables corresponding to primal constraints with a nonzero constant term.
Thus, the dual variables $\beta_{a b x y}$ provide a quantitative MDIEW even after reformulation; this prove the validity of the witness values reported in~{\cite{Verbanis2016}} as lower bounds on the negativity present in the setup (however, with the opposite sign convention).

\section{Relation to other previous works}
\label{App:Relation}

\subsection{Entanglement cost of nonlocal measurements}

We saw in Section~\ref{Sec:Characterizing} that a semiquantum setup is best described as a distributed quantum measurement.
A previous study of the quantity of entanglement present in a nonlocal measurement was done by Bandyopadhyay et al.~{\cite{Bandyopadhyay2009}}.
The authors defined the {\textit{entanglement cost}} $\mathcal{E}_{\text{C}} (\{ \Pi_{a b} \})$, which \ is the minimal amount of entanglement required to perform the measurement (single round) on top of local operations and classical communication.
We note that $\mathcal{E}_{\text{C}}$ is a lower bound on the entanglement $\mathcal{E} \left( \rhoAB \right)$ of the shared quantum state, as $\rhoAB$ can be used to implement $\{ \Pi_{a b} \}$ in our description.

Second, the {\textit{entanglement production capacity}} $\mathcal{E}_{\text{P}} (\{ \Pi_{a b} \})$ is the amount by which the measurement $\{ \Pi_{a b} \}$ can increase the entanglement between A and B, averaged over all measurement outcomes.
The value of $\mathcal{E}_{\text{P}}$ depends on the particular states used as inputs in the recovery protocol, and would have to be maximized over all possibilities.
In Section~\ref{Sec:Characterizing}, we described a particular protocol that produces the states $\mu_{a b}$ with average entanglement $\nu^{\ast}$ in Eq.~(\ref{Eq:LowerBound1}).
As we did not optimize over all recovery protocols, $\nu^{\ast}$ is a lower bound on the entanglement production capacity $\mathcal{E}_{\text{P}}$.
Then, as LOCC operations cannot increase entanglement, the entanglement production capacity is a lower bound on the entanglement cost:
\begin{equation}
 \nu^{\ast} \leqslant \mathcal{E}_{\text{P}} (\{ \Pi_{a b} \}) \leqslant \mathcal{E}_{\text{C}} (\{ \Pi_{a b} \}) \leqslant \mathcal{E} \left( \rhoAB
 \right),
\end{equation}
and the bound holds even if the devices are allowed classical communication.
The question of obtaining better bounds for $\mathcal{E}_{\text{C}}$ remains open.

\subsection{All entangled quantum states are nonlocal}

The construction in Section~\ref{Sec:Characterizing} is closely related to the pioneering paper by Buscemi~{\cite{Buscemi2012}} on semiquantum games.
Specifically, let us look at the Proposition 1 of~{\cite{Buscemi2012}}.
We consider the state $\sigma \in \mathsf{Herm}_+ \left( \mathcal{H}_{\text{A}'} \otimes \mathcal{H}_{\text{B}'} \right)$, and the effective POVM $\overline{\Pi}^{\sigma} = \{ \Pi_{a b} \}$ obtained by using the generalized Bell measurements~(\ref{Eq:GeneralizedBellMeasurement}) on $\sigma$.
All the states $\mu_{a b}$ recovered out of $\overline{\Pi}^{\sigma}$ can be corrected locally to $\sigma$ by Eq.~(\ref{Eq:RecoveredStatesBell}), thus $\overline{\Pi}^{\sigma}$ can be substituted for $\sigma$ in any task.
In the notation of~{\cite{Buscemi2012}}, any state $\rho$ that can realize $\overline{\Pi}^{\sigma}$ majorizes $\sigma$ ($\rho \succcurlyeq_{\text{sq}} \sigma$).
Now, all the states $\rho$ that can be transformed into $\sigma$ by LOSR maps ($\rho \rightarrowtail \sigma$ in~{\cite{Buscemi2012}}) can realize $\overline{\Pi}^{\sigma}$.
This corresponds to the trivial part of the Proposition.

Let us move to the nontrivial part.
Consider the set $\mathcal{P}_{\sigma} = \{ \vec{P} \}$ of correlations $P (a b | x y)$ that can be realized from the state $\sigma$ and shared randomness.
This set is convex by construction.
Let $\rho$ be a generic state and consider the set $\mathcal{P}_{\rho}$ defined similarly.
By convexity, the condition $\rho \succcurlyeq_{\text{sq}} \sigma$ of~{\cite{Buscemi2012}} can be understood as $\mathcal{P}_{\sigma} \subseteq \mathcal{P}_{\rho}$.
When this condition holds, we use the correspondence between the correlations $P (a b | x y)$ and $\{ \Pi_{a b} \}$: the correlations corresponding to $\overline{\Pi}^{\sigma}$ are in $\mathcal{P}_{\sigma}$, and thus in $\mathcal{P}_{\rho}$.
So $\rho$ can realize $\overline{\Pi}^{\sigma}$ and thus $\rho \rightarrowtail \sigma$.

\subsection{MDIEWs for All Entangled Quantum States}

The first constructions of measurement-device-independent entanglement witnesses were presented in~{\cite{Branciard2013,Rosset2013a}}.
These works prescribed that the devices A and B implement partial two-outcome Bell-like measurements with $A_1 = \ketbra{\varphi_{\dX}}$ and $B_1 = \ketbra{\varphi_{\dY}}$, albeit in a fixed basis.
The effective POVM element $\Pi_{11}$ satisfies~(\ref{Eq:GeneralizedBellMeasurement}) and $\nu^{\ast}$ contains the contribution of $p_{1 1} \mathcal{E} (\mu_{1 1}) = \dX^{- 1} \dY^{- 1} \mathcal{E} \left( \rhoAB \right)$.
When $\rhoAB$ is entangled, there exists an entanglement measure such that $\mathcal{E} \left( \rhoAB \right) > 0$, and thus the presence of entanglement can be certified by $\nu^{\ast} > 0$.

\subsection{Measurement-device-independent entanglement and randomness estimation in quantum networks}

The parallel work by Supic et al.~\cite{Supic2017a} also considers the quantification of entanglement in semiquantum scenarios using different measures of robustness and negativity.
Our formulations differ: we abstract over entanglement measures by using entanglement cones, while their work derives explicitly the final semidefinite programs for each of these measures.
However, the resulting optimization problems are mathematically equivalent, and will provide the same lower bounds.
In that sense, both manuscripts are complementary.
Our abstract approach is generic, and provides an operational interpretation of the quantified entanglement as entanglement recovered from the devices, whereas their work provides detailed studies of specific entanglement measures with associated lower-level semidefinite implementations.

Both works also provide additional tools.
The manuscript of Supic et al. also discusses entanglement quantification in the multi-party case and the certification of randomness.
Our approach is oriented towards pratical applications.
In tackling experimental challenges (statistical noise, losses), our conic programs have the merit of being generic, so that these challenges can be solved once for all entanglement measures.
In particular, the easy derivation of quantitative witnesses from the conic dual is invaluable to deal with probability distributions estimated from a finite number of samples, which do not respect the nonsignaling constraints in the presence of statistical noise.

\end{document}